\documentclass[journal]{IEEEtran}
\ifCLASSINFOpdf
\else
\fi
\usepackage{fixltx2e}

\usepackage{dblfloatfix}

\ifCLASSOPTIONcaptionsoff
  \usepackage[nomarkers]{endfloat}
 \let\MYoriglatexcaption\caption
 \renewcommand{\caption}[2][\relax]{\MYoriglatexcaption[#2]{#2}}
\fi

\usepackage{multirow}
\usepackage[latin1]{inputenc} 
\usepackage{tikz}
\usetikzlibrary{shapes,arrows} 
\usepackage[cmex10]{amsmath}
\usepackage{amssymb}
\usepackage{amsmath}
\usepackage{enumerate}

\usepackage[cmex10]{amsmath}
\usepackage{amssymb}
\newcommand{\qed}{\nobreak \ifvmode \relax \else
      \ifdim\lastskip<1.5em \hskip-\lastskip
      \hskip1.5em plus0em minus0.5em \fi \nobreak
      \vrule height0.75em width0.5em depth0.25em\fi}

\usepackage[american,cute inductors,smartlabels]{circuitikz}
\usetikzlibrary{arrows,automata}									
\usetikzlibrary{arrows,shapes,backgrounds,calc,positioning,patterns}
\usetikzlibrary{calc,arrows,shapes,backgrounds,calc,positioning,patterns,decorations.pathmorphing,decorations.markings,mindmap,trees}
\tikzstyle{block} = [draw, rectangle, minimum height=2em, minimum
width=4em] \tikzstyle{sum} = [draw, fill=blue!20, circle, node
distance=1cm] \tikzstyle{input} = [coordinate] \tikzstyle{output} =
[coordinate] \tikzstyle{pinstyle} = [pin edge={to-,thin,black}]
\usetikzlibrary{calc}
\ctikzset{bipoles/thickness=1} \ctikzset{bipoles/length=0.8cm}
\ctikzset{bipoles/diode/height=.375}
\ctikzset{bipoles/diode/width=.3}
\ctikzset{tripoles/thyristor/height=.8}
\ctikzset{tripoles/thyristor/width=1}
\ctikzset{bipoles/vsourceam/height/.initial=.7}
\ctikzset{bipoles/vsourceam/width/.initial=.7} \tikzstyle{every
	node}=[font=\small] \tikzstyle{every path}=[line width=0.8pt,line
cap=round,line join=round]
\usepackage{amsfonts}
\newtheorem{theorem}{Theorem}                						
\newtheorem{lemma}{Lemma}
                   						
\newtheorem{property}{Property}                   						
\newtheorem{problem}{Problem}                   						
\newtheorem{remark}{Remark}
\newtheorem{proposition}{Proposition} 
\usepackage{proof}
\usepackage{algpseudocode}
\usepackage{upgreek}										
\newcommand{\numberset}{\mathbb}							

\newcommand{\R}{\numberset{R}}

\usepackage{mathtools}																	
\usepackage{booktabs}										
\usepackage{graphics} 										
\usepackage{epsfig} 										
\usepackage{contour}										
\usepackage{bm}
\usepackage{times} 		
\usepackage{rotating}										
\usepackage{longtable}					
\usepackage{pdflscape}	
\usepackage{cite} 									
\usepackage{xcolor} 
\usepackage{dsfont}		

\usepackage{centernot}
\usepackage{color}
\usepackage[acronym]{glossaries}
\usepackage[ruled]{algorithm2e}
\usepackage{graphicx}
\usepackage{tabu}
\usepackage[official]{eurosym}

\usepackage{placeins}
\usepackage{siunitx}										
\renewcommand{\vec}{\bm}

\definecolor{darkgreen}{rgb}{0.0, 0.5, 0.0}

\DeclareMathAlphabet{\mathpzc}{OT1}{pzc}{m}{it}
\DeclareMathAlphabet{\mathcal}{OMS}{cmsy}{m}{n}

\newtheorem{assumption}{Assumption}

\newtheorem{objective}{Objective}

\usepackage{cases}
\usepackage{tikz}				
\usetikzlibrary{arrows,shapes,backgrounds,calc,positioning,patterns}

\DeclareMathOperator{\diag}{diag}

\DeclareMathOperator{\blockdiag}{blockdiag}
\DeclareMathOperator{\col}{col}

\usepackage{array}

\usepackage{fixltx2e}

\usepackage{dblfloatfix}

\usepackage{url}

\usepackage{balance}

\usepackage{tikz,pgfplots}
\usepackage{epsfig}

\begin{document}
%
\title{Output Regulation for Load Frequency Control}
%
%
%
%
%
\author{Amirreza~Silani, Michele~Cucuzzella, Jacquelien~M.~A.~Scherpen, Mohammad~Javad~Yazdanpanah
    \thanks{A.~Silani, M.~Cucuzzella, and J.~M.~A.~Scherpen  are with Jan C. Wilems Center for Systems and Control, ENTEG, Faculty of
    	Science and Engineering, University of Groningen, Nijenborgh 4, 9747
    	AG Groningen, The Netherlands
    	{\tt\small \{a.silani, m.cucuzzella, j.m.a.scherpen\}@rug.nl}}%
    \thanks{A.~Silani and M.~J.~Yazdanpanah are with Control \& Intelligent Processing Center of Excellence, School of
    	Electrical and Computer Engineering, University of Tehran, Tehran,
    	Iran 
    	{\tt\small \{a.silani, yazdan\}@ut.ac.ir}}%
    \thanks{This work is supported by the EU Project MatchIT (project number: 82203)}

 } 

\maketitle

\begin{abstract}
Motivated by the inadequacy of the existing control strategies for power systems affected by \emph{time-varying} uncontrolled power injections such as loads and the  increasingly widespread renewable energy sources, 
this paper proposes two control schemes based on the well-known output regulation control methodology. The first one is designed based on the classical output regulation theory and addresses the so-called Load Frequency Control (LFC) problem in presence of \emph{time-varying} uncontrolled power injections. Then, in order to also minimize the generation costs, we use an approximate output regulation method that solves numerically only the partial differential equation of the regulator equation and propose a controller based on this solution, minimizing an appropriate penalty function. 
An extensive case study shows excellent performance of the proposed control schemes in different and critical scenarios.
\end{abstract}

%
\IEEEpeerreviewmaketitle

\section{Introduction}
In power networks, the supply-demand mismatch induces frequency deviations from the nominal value, eventually leading to fatal stability disruptions~{\cite{ref1,ref2}}. 
Therefore, reducing this deviation is of vital importance for the overall network resilience and reliability, attracting a considerable amount of research activities on the design and analysis of the so-called {Load Frequency Control (LFC), also known as Automatic Generation Control (AGC), where a suitable control scheme continuously changes the generation setpoints to compensate supply-demand mismatches}, regulating the frequency to the corresponding nominal value (see for instance~{\cite{ref1,ref2}} and the references therein). Moreover, besides ensuring the stability of the overall power infrastructure, in order to solve the so-called economic dispatch problem {\cite{ref47}}, modern control schemes aim also at  reducing the operational costs associated to the LFC. In the literature (see for instance~{\cite{ref47,ref20,ref21,ref49}} and the references therein), this control objective is referred to as \textit{Optimal} LFC~(OLFC). However, as a consequence of the increasingly share of renewable energy sources, we are not sure if the existing control systems are still adequate {\cite{ref3}}. 

\subsection{Literature Review}
Traditionally, a power network is subdivided in the so-called \emph{control areas}, each of which represents an electric power system or combination of electric power systems to which a common LFC scheme is applied {\cite{Ibraheem,Pandey}}.
The LFC problem is usually addressed at each control area by primary and secondary control schemes. More precisely, the primary control layer preserves the stability of the power system acting faster than the secondary control layer, which typically provides the generation setpoints to each control area {\cite{ref47}}. 
Then, in order to obtain OLFC, a tertiary control layer can be used to reduce the generation costs in slow timescales. To tackle the same problem in fast timescales, distributed control schemes are usually adopted, where the control areas cooperate with each other {\cite{ref41}}. 
For the latter case, there exist generally two types of control approaches: consensus-based protocols or primal-dual algorithms. By using the first approach, all the control areas that exchange information through a communication network achieve the same marginal cost, solving the OLFC problem typically in absence of constraints   \cite{ref24,ref25,ref26,ref27,ref29,ref30,ref33,ref34,ref35,ref16,ref17}. The second approach performs OLFC by solving an optimization  problem  that may potentially include constraints for instance on the generated or exchanged power {\cite{ref20,ref21,ref11,ref12,ref13,ref14,ref15,ref18,ref19,ref22,ref23,ref28,ref31}}. 
In the following, we briefly discuss some of the relevant works in the literature on the design and analysis of control schemes achieving LFC and OLFC.


In \cite{robust,ref39,ref40,ref42}, different control schemes for solving the LFC problem in presence of constant loads are proposed. More precisely, in \cite{ref39} and \cite{ref40}, distributed PI droop controllers are designed, while stability conditions for droop controllers are investigated in \cite{ref42}, where the well-known port-Hamiltonian framework is used. Based on the sliding mode control methodology, a decentralised control scheme is proposed in \cite{robust}, where besides frequency regulation, the power flows among different areas are maintained at their scheduled values.
In \cite{ref41,ref11,ref44,ref46,ref26,ref47,ref48,ref49,ref50}, different control schemes for solving the OLFC problem in presence of constant loads are proposed. More precisely, a distributed passivity-based control scheme is proposed in \cite{ref47}, where the voltages are assumed to be constant.
A distributed sliding mode control strategy is proposed in \cite{ref49}, where,  although the robustness property of sliding mode is able to face \emph{time-varying} loads, the stability of the desired equilibrium point is established under the assumption of constant loads only.
A hierarchical control scheme is proposed in \cite{ref41}, while decentralized  integral control and distributed averaging-based integral control schemes are proposed in \cite{ref26}. 
In \cite{ref11}, the convergence is proved under the assumptions of convex cost functions and known power flows, while a gradient-based approach is proposed in \cite{ref44}. 
A linearized power flow model is adopted in \cite{ref46}, while a primal-dual approach is proposed in \cite{ref48}, where an aggregator collects the frequency measurements from all the control areas in order to compute and broadcast the generation setpoints to each control area. A real-time bidding mechanism is developed in~\cite{ref50}. 
{In \cite{antonella1}, a distributed sliding mode observer-based scheme is proposed  to estimate the frequency deviation and perform robust fault reconstruction. In \cite{antonella2}, higher order sliding mode observers are presented to robustly and dynamically estimate the unmeasured state variables in power networks.}

Before presenting in the next subsection the motivations and contributions of this paper, we notice that in all the above mentioned works on LFC and OLFC, theoretical guarantees are established under the common assumption that loads are constant.

\subsection{Motivation and Contributions}

Nowadays, renewable energy sources and \emph{new} loads such as plug-in electric vehicles are an integral part of the power infrastructure. As a consequence, unavoidable uncertainties are sharply increasing and may put a strain on the system stability. For this reason, the resilience and reliability of the power grid may benefit from the design and analysis of control strategies that theoretically guarantee the system stability in presence of \emph{time-varying} loads and renewable sources. 
To do this,  an internal model approach is proposed in \cite{ref16}, where the loads behaviour is described as the output of a dynamical exosystem, as it is customary in output regulation theory~{\cite{ref49.2,Isidori}}.
However, in \cite{ref16} the turbine governor dynamics are neglected, while it is generally important in terms of tracking performance to describe the generation side in a satisfactory level of detail. Moreover, the exosystem model adopted in \cite{ref16} to describe the load dynamics is linear, assumed to be incrementally passive and generally does not allow to achieve OLFC.
Differently from~\cite{ref16}, we first design a controller based on the classical output regulation theory~{\cite{ref49.2}} to solve the LFC problem for a class of exosystems wider than the one adopted in \cite{ref16}. More precisely, in this paper we describe the behaviour of the uncontrolled power injections, \textit{i.e.}, the difference between the power generated by the renewable energy sources and the one absorbed by the loads, by nonlinear exosystems that do not need to be incrementally passive and are independent from the system parameters. 
Secondly, in order to also minimize the generation costs, we use an approximate output regulation method that solves numerically only the partial differential equation of the regulator equation and propose a controller based on this solution, achieving an \emph{approximate} OLFC, where the norm of the error (\textit{i.e.}, the frequency deviation from the nominal value and the difference between the actual generated power and its corresponding optimal value) is upperbounded by a sufficiently small positive constant, whose influence on the performance of the controlled system is shown (by an extensive simulation analysis)  to be negligible in practical applications.
{Moreover, we provide an elegant control design procedure that provides conditions for the solvability of the regulator equations, reduces the computational burden for solving the regulator equations and guarantees stability also in presence of nonlinear time-varying loads and renewable generation.}

The main contributions of this paper can be summarized as follows: 
\begin{itemize}
	\item [(i)]
the LFC problem for nonlinear power networks including \emph{time-varying} uncontrolled power injections is formulated as a standard output regulation problem; 
\item [(ii)]
the {\emph{time-varying}} uncontrolled power injections are represented by the outputs of nonlinear dynamical exosystems;
\item [(iii)] 
we propose a control scheme based on the classical output regulation theory for solving the conventional LFC problem in presence of \emph{time-varying} uncontrolled power injections while ensuring the stability of the overall network;
\item [(iv)]
we use an approximate output regulation method for solving an \emph{approximate} OLFC problem in presence of \emph{time-varying} uncontrolled power injections while ensuring the stability of the overall network.
\end{itemize}
\subsection{Outline}
 This paper is organized as follows. The control problem is formulated in Section~II. The controllers based on the classical output regulation and approximate output regulation are designed and analyzed in Sections~III and~IV, respectively. The simulation results are presented and discussed in Section~V, while in Section~VI conclusions are gathered. 
 \subsection{Notation}
The set of real numbers is denoted by $\R$. The set of positive (nonnegative) real numbers is denoted by $\R_{>0}$ ($\R_{\geq0}$).
Let $\boldsymbol{0}$ denote the vector of all zeros and the null matrix of suitable dimension(s), and $\boldsymbol{1}$ denote the vector containing all ones. The $n \times n$ identity matrix is denoted by $\mathds{I}_n$. Let $A \in \R^{n \times n}$ be a matrix. In case $A$ is a positive definite (positive semi-definite) matrix, we write $A > \vec{0}$  ($A \geq \vec{0}$). Let $\det(A)$ denote the determinant of matrix $A$ and $\vert A\vert$ denote the matrix $A$ with all elements positive. The $i$-th element of vector $x$ is denoted by $x_i$. A steady-state solution to system $\dot x = f(x)$, is denoted by $\overline x$,  \textit{i.e.}, $\boldsymbol{0} = f(\overline x)$. Let $x\in\R^n,y\in\R^m$ be vectors and $W\in\R^{n\times n}$ be a positive semi-definite matrix, then we define $\col(x,y):=(x^\top~y^\top)^\top\in\R^{n+m}$ and $\big\Vert x\big\Vert_{W}:=x^\top W x \geq 0$. Consider the functions $g : \R^n\rightarrow\R^{n\times m}, h : \R^n\rightarrow\R^{n},$ then the Lie derivative of $h(x)$ along $g(x)$ is defined as $L_g h(x):=\frac{\partial h(x)}{\partial x}g(x)$ with $\frac{\partial h(x)}{\partial x}=\col\big(\frac{\partial h_1(x)}{\partial x},\dots,\frac{\partial h_n(x)}{\partial x}\big)$ and $\frac{\partial h_i(x)}{\partial x}=\big(\frac{\partial h_i(x)}{\partial x_1}~\dots~\frac{\partial h_i(x)}{\partial x_n}\big)$ for $i=1,\dots,n$. Given a vector $x\in \R^n$, $[x]\in \R^{n\times n}$ indicates the diagonal matrix whose diagonal entries are the components of $x$ and $\sin(x):=\col\big(\sin(x_1),\dots,\sin(x_n)\big)$. A continuous function $\alpha : \R_{>0} \rightarrow \R_{>0}$ is said to be of class $\mathcal{K}$ if it is nondecreasing and $\alpha(0) = 0$. The bold symbols indicate the solutions to a partial differential equation.  
For notational simplicity, the dependency of the variables on time $t$ is mostly omitted throughout the paper.
\section{Problem Formulation}
In this section, we introduce the nonlinear power system model together with the dynamics of the uncontrolled power injections (\textit{i.e.}, the difference between the power generated by the renewable energy sources and the one absorbed by the loads), which are described as the output of a nonlinear dynamical exosystem. Then, two control objectives are presented: Load Frequency Control (LFC)  and   \emph{approximate Optimal} LFC. 
\subsection{Power Network Model}
In this subsection, we discuss the model of the considered power network (see Table~\ref{tab1} for the description of the symbols and parameters used throughout the paper). 
The  network topology is represented by an undirected and connected graph $\mathcal{G}=(\mathcal{V},\mathcal{E})$, where $\mathcal{V}=\{1, 2, ..., n\}$ is the set of the control areas and $\mathcal{E}=\{1, 2, ..., m\}$ is the set of the transmission lines.
Then, let  $\mathcal{A}\in \mathbb{R}^{n\times m}$ denote the corresponding incidence matrix and let the ends of the transmission line $j$ be arbitrarily labeled with a `$+$' and a  `$-$'. Then, we have
\begin{equation}\label{eq1}
\mathcal{A}_{ij}=\begin{cases}
+1 \quad &\text{if $i$ is the positive end of $j$}\\
-1 \quad &\text{if $i$ is the negative end of $j$}\\
~~0 \quad &\text{otherwise}.
\end{cases}
\end{equation}
Moreover, in analogy with \cite{ref16,ref17}, we assume that the power network is lossless and each node represents an aggregated area of generators and loads. Then, the dynamics (known as \emph{swing} dynamics) of node (area) $i\in\mathcal{V}$ are the following (see also \cite{ref16,ref17,ref47,ref49} for further details):
\begin{equation}\label{eq2}
\begin{split}
\dot{\varphi} _i=&~\omega_{i} \\
\tau _{p i} \dot{\omega}_{i}=&-\psi _{i}\omega_{i} +P_{c i}+P_{d i}+\sum _{j\in \mathcal{N}_i} V_i V_j B_{ij}\sin(\varphi _i - \varphi _j)\\
\tau_{v i}\dot{V}_i=&~\bar{E}_{f i}-(1-\chi_{d i}B_{ii})V_i\\
&+\chi_{d i}\sum_{j\in \mathcal{N}_i}V_jB_{ij}\cos(\varphi _i - \varphi _j),
\end{split}
\end{equation}
where $\varphi_i,\omega_i, V_i,   P_{d i} : \R_{\geq 0}\rightarrow \R$, $\tau _{p i}, \tau _{v i}, \psi _{i},\bar{E}_{f i},B_{ii}, $ $B_{ij} \in \R,  \chi _{d i}:=X _{d i}-X' _{d i}$, with $X _{d i}, X' _{d i}\in\R$, and $P_{c i}: \R_{\geq 0}\rightarrow \R$ is the power generated by the $i$-th (equivalent) synchronous generator and can be expressed as the output of a first-order dynamical system describing the behavior of the turbine-governor {\cite{ref47,ref49}}, \textit{i.e.},
\begin{equation}\label{eq4.1}
	\begin{split}
		\tau_{ci}\dot{P}_{ci}=-P_{ci}-\xi_{i}^{-1}\omega_{i} +u_{i},
	\end{split}
\end{equation}
where $u_{i} : \R_{\geq 0}\rightarrow \R$ is the control input and $\tau_{c i},\xi_i\in\R$. 

\begin{table}
	\centering
	\caption{Symbols}
	{\begin{tabular}{ll}
			\toprule			

			$P_{c i}$		& Conventional power generation\\
			$P_{d i}$		& Uncontrolled power injection\\  
			$\varphi _i$	& Voltage angle\\
			$\omega_i$	    & Frequency deviation\\
			$V_i$ 	    	& Voltage \\
			$d_{ai}, d_{bi}$ 	& States of the exosystem describing the uncontrolled power injections \\
			$\tau _{p i}$   & Moment of inertia\\
			$\tau _{v i}$   & Direct axis transient open-circuit constant\\
     		$X_{d i}$       & Direct synchronous reactance\\
			$X'_{d i}$      & Direct synchronous transient reactance\\
			$\psi_i$        & Damping constant\\
			$B$             & Susceptance\\
			$\bar{E}_{f i}$       & Constant exciter voltage\\
     	 	$\mathcal{N}_i$ & Neighboring areas of area $i$\\
     	 	$\tau_{c i}$    & Turbine time constant\\
     	 	$\xi_i$         & Speed regulation coefficient\\
     	 	$\Gamma_i$        & Constant parameter of uncontrolled power injections\\       
     	 	$\mathcal{A}$        & Incidence matrix of power network\\    
     	 	$L^{\mathrm{com}}$   & Laplacian matrix of communication network\\ 
			$u_i$	     	& Control input \\
			\bottomrule
	\end{tabular}}
	\label{tab1}
	\vspace{-1em}
\end{table} 
Now we can write systems (\ref{eq2}) and (\ref{eq4.1}) compactly for all nodes $i \in \mathcal{V}$ as 
\begin{equation}\label{eq3}
\begin{split}
\dot{\theta}&=\mathcal{A}^\top \omega \\
\tau _p \dot{\omega} &= -\psi \omega +
P_c+P_d -\mathcal{A}\Upsilon(V) \sin(\theta)\\
\tau _v \dot{V}&=-\chi _dE(\theta)V+\bar{E}_{f}\\
\tau_{c}\dot{P}_{c}&=-P_{c}-\xi^{-1}\omega +u,
\end{split}
\end{equation}
where $\omega, V, P_c, P_d, u : \R_{\geq 0}\rightarrow \R^{n}$,  $\theta : \R_{\geq 0}\rightarrow \R^{m}$ denotes the vector of the voltage angles differences, \textit{i.e.}, $\theta:=\mathcal{A}^\top\varphi$, $\tau _p, \tau _v, \psi, \chi_d, \tau _c, \xi \in \R^{n\times n}$, and $\bar{E}_f\in\R^n$. Moreover, $\Upsilon : \R^{n}\rightarrow \R^{m\times m}$ is defined as $\Upsilon(V):=\diag\{\Upsilon_1, \Upsilon_2, ..., \Upsilon_m\}$, with $\Upsilon_k:=V_iV_jB_{ij}$, where $k\sim\{i,j\}$ denotes the line connecting areas $i$ and $j$. Furthermore, for any $i,j\in \mathcal{V}$, the components of $E : \R^m\rightarrow \R^{n\times n}$ are defined as follows: 
 \begin{align}
   \begin{split}
E_{ii}(\theta)=&~\frac{1}{\chi_{d i}}-B_{ii}, \hspace{6.8em} i \in \mathcal{V}\\
E_{ij}(\theta)=&-B_{ij}\cos(\theta_k)=E_{ji}(\theta), \hspace{1.2em} k \sim \{i,j\} \in \mathcal{E} \\
E_{ij}(\theta)=&~0,  \hspace{10.35em}\text{otherwise}.
   \end{split}
 \end{align}
\begin{remark}\textbf{(Susceptance and reactance).}
	 According to \cite[Remark~1]{ref49}, we notice that the reactance $X_{di}$ of each generator $i\in\mathcal{V}$ is in practice generally larger than the corresponding transient reactance $X'_{di}$. Furthermore, the self-susceptance $B_{ii}$ is negative and satisfies $\vert B_{ii}\vert>\sum_{j\in \mathcal{N}_i} \vert B_{ij}\vert$. Therefore, $E(\theta)$ is a strictly diagonally dominant and
symmetric matrix with positive elements on its diagonal, implying that $E(\theta)$
is positive definite \cite{ref16}.
\end{remark}

Now, as it is customary in the power systems literature (see for instance \cite{ref16,ref47,ref49}), we assign to the conventional power generation, \textit{i.e.}, the power generated by the synchronous generators, the following strictly convex linear-quadratic cost function to model the costs of the conventional power generations\footnote{{In this work we consider only uncontrollable loads. Thus, we do not include in the cost function \eqref{eq4.61} the cost associated with the load adjustment.}}:  
\begin{equation}\label{eq4.61}
\begin{split}
J(P_c)=P_c^\top QP_c+R^\top P_c+\boldsymbol{1}_n^\top C,
\end{split}
\end{equation}
where $J : \R^n\rightarrow \R$, $Q \in\R^{n\times n}$, and $ R, C \in\R^{n}$.
 Furthermore, to permit the design of the control input $u$ in Sections \ref{sec:COR} and \ref{sec:AOR}, as a first step, we augment system \eqref{eq3} with the following distributed dynamics \cite{ref47,ref49}:
\begin{equation}\label{eq4.6}
\begin{split}
	\tau_\delta \dot{\delta}&=-\delta +P_c-\xi^{-1}QL^{\mathrm{com}}(Q\delta + R),\\
\end{split}
\end{equation}
where $\delta : \R_{\geq 0}\rightarrow \R^{n}$, $\tau_{\delta} \in\R^{n\times n}$, and  $L^{\mathrm{com}}\in\R^{n\times n}$ is the Laplacian matrix associated with a communication network, whose corresponding graph is assumed to be undirected and connected. More precisely, the term $Q\delta + R$ in \eqref{eq4.6} reflects the (virtual) marginal cost associated with the conventional power generation and $L^{\mathrm{com}}(Q\delta + R)$ represents the exchange of such information among the areas of the power network.
 In Section~III, we will show that \eqref{eq4.6} plays an important role in the stability analysis of system \eqref{eq3}. 
More precisely, in Section III the state variable $\delta$ will be used to design a state-feedback controller that stabilizes system \eqref{eq3}. Then, the gain of such controller will be used for the output regulation control methodology.

\subsection{Exosystem Model}
The power associated with the renewable generation and load demand is in practice \emph{time-varying}. The renewable generation depends indeed on several exogenous factors such as the wind speed and solar radiation for wind and photovoltaic energy, respectively.  
Also, the load demand generally depends on several exogenous factors such as the weather conditions, usage patterns and social aspects \cite{ref49.1}. Consequently, the \emph{time-varying} behaviour of loads and renewable generation sources can be described by dynamic systems that are independent from the state of the power network \cite{ref49.1}.
{The class of exosystems we consider in this paper are capable to accurately describe the real behaviour of loads and renewable energy sources and are often adopted in the literature \cite{ref16,Wang,Arif}. In practice, the parameters of the exosystems can be identified from historical data \cite{Choi,Ge,Hill}, {as we will do in Section \ref{sec:sim} (see Scenario~3)}}.
Then, as it is customary in output regulation control theory \cite{ref49.2,Isidori}, we consider  the uncontrolled power injections (\textit{i.e.}, the difference between the power generated by the renewable energy sources and the one absorbed by the loads) as exosystems and the dynamics of the $i$-th uncontrolled power injection can be expressed as follows\footnote{{The adopted exosystem is a general nonlinear system capable of reproducing a large class of signals and may potentially also take into account the presence of storage devices.}} (see for instance \cite{ref16} and the references therein): 
\begin{equation}\label{eq4.2}
\begin{split}
\dot{d}_{a i}&=0 \\
\dot{d}_{b i}&=s_{i}(d_{b i}) \\
P_{d i}&=\Gamma_{i} \begin{pmatrix}
d_{a i} \\
d_{b i}
\end{pmatrix},
\end{split}
\end{equation}
where ${d}_{a i} : \R_{\geq 0}\rightarrow \R$, {${d}_{b i} : \R_{\geq 0} \rightarrow \R^{n_d}$} are the states of the exosystem describing the constant and \emph{time-varying} components of  the uncontrolled power injection $i\in\mathcal{V}$, respectively, {$s_{i}:\R^{n_d}\rightarrow\R^{n_d}$} {is Lipschitz,} and {$\Gamma_{i}\in\R^{1\times (n_d+1)}$}. 
Then, (\ref{eq4.2}) can be written compactly for all nodes $i \in \mathcal{V}$ as
\begin{equation}\label{eq4.3}
\begin{split}
\dot{d}&=S(d) \\
P_d&=\Gamma d, 
\end{split}
\end{equation}
where {$d: \R_{\geq 0} \rightarrow \R^{n(n_d+1)}$} is defined as $d := \col(d_{a1}, d_{b1}, \dots, d_{an}, d_{bn}), P_d : \R_{\geq 0}\rightarrow \R^{n}$, {$S: \R^{n(n_d+1)} \rightarrow \R^{n(n_d+1)}$} is defined as $S := \col(0, s_{1},\dots,0, s_{n})$, and  {$\Gamma := \blockdiag(\Gamma_{1}, \dots,  \Gamma_{n})\in\R^{n\times n(n_d+1)}$}. {Although we treat the constant component as an uncertain term (see the first line in \eqref{eq4.2}), the estimated time-varying components may differ in practice from the actual ones. However, the theoretical analysis in presence of uncertain exosystems is out of the scope of this work and we leave as future research the possibility of designing a controller based on the robust output regulation methodology presented in \cite[Chapters~6, 7]{ref49.2}. Although we do not tackle this aspect theoretically, we will show in simulation (see Section \ref{sec:sim}, Scenario 3) that the controlled system is input-to-state stable with respect to the possible mismatch between the actual power injection and the one generated by the corresponding exosystem.}


\begin{remark}\textbf{(Nonlinear exosystem).}\label{remark exo}
Note that the load demand considered in \cite[Corollary~2]{ref16} is given by the  superposition of three different terms: (i) a constant, (ii) a periodic component that can be compensated optimally and (iii) a periodic component that cannot be compensated optimally. The corresponding exosystem is linear, assumed to be incrementally passive and depends on some predefined constant matrices.
Differently, the considered  exosystem \eqref{eq4.3} can be generally nonlinear and does not depend on any predefined parameter. 
\end{remark}
%
\subsection{Control Objectives}
In this subsection, we introduce and discuss the main control objectives of this work.
 The first objective concerns the asymptotic regulation of the  frequency deviation to zero, \textit{i.e.}, 
\begin{objective}\textbf{(Load Frequency Control).}\label{objective1}
	\begin{equation}\label{obj1}
	\lim_{t\rightarrow \infty}\omega(t)=\boldsymbol{0}_n.
\end{equation}
\end{objective}
Besides improving the stability of the power network by regulating the frequency deviation to zero, advanced control strategies additionally aim at reducing the costs associated with the power generated by the conventional synchronous generators.
In this regard, \cite[Lemma 2]{ref49}, \cite[Lemma~3]{ref16} show that it is possible to achieve zero steady-state frequency deviation and simultaneously minimize the total generation cost function~\eqref{eq4.61} when the uncontrolled power injection $P_d$ is constant. More precisely, when the uncontrolled power injection $P_d$ is constant, the \emph{optimal} value of $P_c$, which allows for zero steady-state frequency deviation and minimizes (at the steady-state) the total generation cost \eqref{eq4.61}, is given by:
\begin{equation}\label{optimal}
P_c^\mathrm{opt}= Q^{-1}\Big(\frac{\boldsymbol{1}_n\boldsymbol{1}_n^\top (Q^{-1}R-P_d)}{\boldsymbol{1}_n^\top Q^{-1}\boldsymbol{1}_n}-R\Big),
\end{equation}  
{which is the solution to the following optimization problem}
{\begin{equation}\label{minimization}\nonumber
\begin{split} 
&\min~J(P_c)\\ 
&~\text{s.t.}~\boldsymbol{1}_n^\top (\bar{u}+P_d)=0,
\end{split}
\end{equation}}
{where $J(P_c)$ is given in \eqref{eq4.61} (see \cite{ref49,ref16} for further details).}
This leads us to the second objective concerning the frequency regulation and minimization of the total generation cost, which is also known in the literature as economic dispatch or \emph{Optimal} LFC (OLFC) \cite{ref16,ref49}. More precisely, in \cite[Subsection 6.1]{ref16} it is shown that economically efficient frequency regulation can be achieved in the presence of a particular class of linear \emph{time-varying} load models (see Remark~\ref{remark exo} for more details on the exosystem model adopted in \cite{ref16}).
However, the achievement of OLFC in presence of a wider class of nonlinear \emph{time-varying} uncontrolled power injections appears complex and still challenging. 
Then, in order to address this challenging task in Section \ref{sec:AOR}, we introduce now the concept of \emph{approximate} OLFC (\mbox{$\epsilon-$OLFC}), \textit{i.e.},
\begin{objective}\textbf{(Approximate Optimal Load Frequency Control).}	\label{objective2}
\begin{equation}\label{obj2}
\lim_{t\rightarrow \infty}\Vert\col\big(\omega(t), P_c(t)\big)-\col\big(\boldsymbol{0}_n, P_c^\mathrm{opt}(t)\big)\Vert \leq \epsilon,
\end{equation}	
\end{objective}
with $\epsilon\in\R_{\geq 0}$ and $P_c^\mathrm{opt}(t)$ given by \eqref{optimal}.

We notice that when $\epsilon=0$, Objective 2 becomes identical to the classical  OLFC objective, \textit{i.e.,} $\lim_{t\rightarrow \infty}\col\big(\omega(t), P_c(t)\big)=\col\big(\boldsymbol{0}_n, P_c^\mathrm{opt}(t)\big)$. 

We assume now that there exists a (suitable) steady-state solution to the considered augmented power network model (\ref{eq3}), (\ref{eq4.6}) and (\ref{eq4.3}).  
\begin{assumption}\textbf{(Steady-state solution).}\label{ass1}
	There exist a constant input $\bar{u}$ and a steady-state solution $(\bar{\theta}, \bar{\omega}, \bar{V}, \bar{P}_c, \bar{\delta}, \bar{d})$ to (\ref{eq3}), (\ref{eq4.6}) and (\ref{eq4.3}) satisfying 	
	\begin{equation}\label{eq4}
	\begin{split}
	\boldsymbol{0}&=\mathcal{A}^\top\bar{\omega} \\
	\boldsymbol{0}&= -\psi\bar{\omega} +\bar{P}_c+\Gamma \bar{d}-\mathcal{A}\Upsilon(\bar{V}) \sin(\bar{\theta})\\
	\boldsymbol{0}&=-\chi_d E(\bar{\theta})\bar{V}+\bar{E}_{f}\\
	\boldsymbol{0}&=-\bar{P}_c-\xi^{-1}\bar{\omega}+\bar{u}\\
	\boldsymbol{0}&=-\bar{\delta} +\bar{P}_c-\xi^{-1}QL^{\mathrm{com}}(Q\bar{\delta} + R)\\
	\boldsymbol{0}&=S(\bar{d}). \\
	\end{split}
	\end{equation}
Additionally, \eqref{eq4} holds also when $\bar{\omega}=\boldsymbol{0}$ and $\bar{P}_c = P_c^\mathrm{opt}$, with $P_c^\mathrm{opt}$ given by \eqref{optimal}.
\end{assumption}

In the next section, we formulate the LFC problem as a classical output regulation control problem and design a control scheme for regulating the frequency in presence of \emph{time-varying} loads and renewable generation sources. To this end, we will first consider the state-feedback controller presented
in \cite{ref49}, 
which guarantees the asymptotic stability of system \eqref{eq3} and  \eqref{eq4.6} with \textit{constant} uncontrolled power
injections. 
As a consequence, in analogy with \cite{ref49,ref16}, the following assumption is required:
\begin{assumption}\textbf{(Steady-state voltage angle and amplitude).}\label{ass2}
	The steady-state voltage $\bar{V}\in\R^n$ and angle difference $\bar{\theta}\in\R^m$ satisfy 
	\begin{equation}\label{eq5}
	\begin{split}
	&\bar{\theta}_l \in (-\frac{\pi}{2},\frac{\pi}{2}),~\forall l\in\mathcal{E}\\
    \dfrac{1}{\chi_{d i}}-B_{ii}+\sum_{l\sim{i,j}\in\mathcal{E}}&\dfrac{B_{ij}\big(\bar{V}_i+\bar{V}_j\sin^2(\bar{\theta}_l) \big)}{\bar{V}_i\cos(\bar{\theta}_l)}>0,~   \forall i\in\mathcal{V}.
	\end{split}
	\end{equation}
\end{assumption}
Note that Assumption~\ref{ass2} is usually verified in practice, \textit{i.e.}, the differences in voltage
(angles) are small and the line reactances are greater than the
generator reactances \cite{ref49,ref16}.

\section{Output Regulation for Load Frequency Control}
\label{sec:COR}
In this section, we use the output regulation control methodology introduced in \cite{ref49.2} to achieve Objective~1, while the design of a control algorithm achieving Objective~2 is addressed in Section~\ref{sec:AOR}.

First, let the state variable $x : \R_{\geq 0}\rightarrow \R^{m+4n}$ be defined as  $x:=\col(\theta,\omega,V,P_c,\delta)$, {$d : \R_{\geq 0}\rightarrow \R^{n(n_d+1)}$} be the exosystem state variable, and $u : \R_{\geq 0}\rightarrow \R^{n}$ be the control input. Then, we can rewrite  (\ref{eq3}), \eqref{eq4.6} and (\ref{eq4.3}) as 
\begin{subequations}\label{eq6}
\begin{align}\label{eq6a}
\dot{x} =&~f(x,d)+g(x,d)u \\ \label{eq6b}
\dot{d} =&~S(d) \\ \label{eq6c}
h(x,d) =&~\omega,
\end{align}
\end{subequations}
where $h(x,d)$ is the output mapping, $g(x,d):=\col(\boldsymbol{0}_{m\times n}, \boldsymbol{0}_{n\times n}, \boldsymbol{0}_{n\times n}, \tau_c^{-1}, \boldsymbol{0}_{n\times n})$ and 
\begin{equation}\label{eq6.1}
\begin{split}
&f(x,d):=
\left(\begin{array}{c}
\mathcal{A}^\top \omega \\
\tau_p^{-1}\Big(-\psi \omega +P_c+\Gamma d-\mathcal{A}\Upsilon(V) \sin(\theta)\Big)\\
\tau_{v}^{-1} (-\chi _d E(\theta)V+\bar{E}_{f})\\
\tau_{c}^{-1}(-P_c-\xi^{-1}\omega) \\
\tau_{\delta}^{-1}(-\delta +P_c-\xi^{-1}QL^{\mathrm{com}}(Q\delta + R))
\end{array}\right).
\end{split}
\end{equation}
Now, we compute the relative degree of system \eqref{eq6}. Indeed, in the following subsections, the relative degree of system \eqref{eq6} is useful to design the controller and simplify the regulator equation. More precisely, it is used to compute the zero dynamics of system \eqref{eq6} and, consequently, reduce the order of the regulator equation, simplifying the output regulation problem.  Let us also define
{\begin{equation}\label{f_a}
\begin{split}
f_a(x,d)&:=\col(f(x,d),S(d))\\
g_a(x,d)&:=\col(g(x,d),\boldsymbol{0}_{n(n_d+1)\times n}),
\end{split}
\end{equation}}
then, according to \cite[Definition~2.47]{ref49.2}, the relative degree of the system \eqref{eq6} is computed in the following lemma. 
\begin{lemma}\textbf{(Relative degree of \eqref{eq6}).} \label{lemma 1}
		For each $i=1,\dots,n$, the $i$-th output $h_i$ of system \eqref{eq6} has relative degree $r_i=2$ for all the trajectories $(x,d)$.
\end{lemma}
\begin{proof}
{See Subsection~\ref{Appendix A} in the Appendix.}
\end{proof}


\begin{remark}\textbf{(Asymptotic stability of system \eqref{eq6a} with constant power injections).}\label{remark2}
	From \cite[Theorem~1, Remark~3]{ref49},  it follows that, if Assumptions \ref{ass1} and \ref{ass2} hold, the state-feedback controller $u=\delta$ asymptotically stabilizes system \eqref{eq6a} when the uncontrolled power injections are assumed to be constant, achieving Objectives 1 and 2. Note that, this result is needed for the solvability of the output regulation problem we introduce in the next subsection (see \cite[Assumption 3.2]{ref49.2} for further details). 
\end{remark}

In the following, we briefly recall for the readers' convenience some concepts of the output regulation control methodology. Then, we design a control scheme achieving Objective~1 in presence of \emph{time-varying} uncontrolled power injections.
\subsection{Output Regulation Methodology}
In analogy with \cite[Section~3.2]{ref49.2}, we first define the nonlinear output regulation problem for system \eqref{eq6} as follows:
\begin{problem}\textbf{(Output regulation).}\label{problem1}
Let Assumptions \ref{ass1}, \ref{ass2} hold and the initial condition  $\big(x(0),d(0)\big)$ of system \eqref{eq6} be sufficiently close to the equilibrium point $(\bar{x},\bar{d})$ satisfying \eqref{eq4}. Then, design a state feedback controller
\begin{equation}\label{uuu}
u(t)=k\big(x(t),d(t)\big)
\end{equation}
such that the closed-loop system \eqref{eq6}, \eqref{uuu} has the following properties:

\begin{enumerate}[]
	\item \begin{property}\label{p1}
		The trajectories $\col\big(x(t),d(t)\big)$ of the closed-loop system exist and are bounded for all $t\geq0$;
\end{property}
\item \begin{property}\label{p2}
	The trajectories $\col\big(x(t),d(t)\big)$ of the closed-loop system satisfy $\lim _{t\rightarrow\infty} h(x,d)=\boldsymbol{0}_n$, achieving Objective~1.
\end{property}
\end{enumerate}
\end{problem}
We say that the (local) nonlinear output regulation problem (Problem~\ref{problem1}) is \textit{solvable}
if there exists a controller such that the closed-loop system satisfies Properties~\ref{p1} and \ref{p2}. Now, in analogy with \cite[Assumption~3.1$^\prime$]{ref49.2}, we introduce the following assumption on the exosystem \eqref{eq6b}.
\begin{assumption}\textbf{(Stability of the exosystem \eqref{eq6b}).}\label{ass3}
The equilibrium $\bar{d}$ of the exosystem \eqref{eq6b} is  Lyapunov stable and there exist an open neighborhood $\mathcal{D}$ of $d=\bar{d}$, where every point is Poisson stable in the sense described in~\cite[Remark~3.2]{ref49.2}.	
\end{assumption}

{Note that the above assumption is the only assumption on the exosystem \eqref{eq6b}, concerning the stability of its equilibrium point. Indeed, the stability of the exosystem, which implies the boundedness of the disturbances, is a standard assumption in several control methodologies since it is very hard to propose a control scheme guaranteeing the stability of  the closed-loop system with unbounded disturbances.  
}
Also, the above assumption is required for establishing the necessary condition for the solvability of Problem~\ref{problem1}, 
which is established in the following theorem.
\begin{theorem}\textbf{(Solvability and regulator equation}\textbf{).} \label{th1}
Let Assumptions~\ref{ass1}--\ref{ass3} hold, then Problem~\ref{problem1} is solvable if and only if there exist smooth functions  $ \boldsymbol{x}(d)$ and $\boldsymbol{u}(d)$ defined for $d\in\mathcal{D}$ such that
\begin{subequations}\label{eq25}
\begin{align}\label{eq25a}
\dfrac{\partial \boldsymbol{x}(d)}{\partial d}S(d) =&~f(\boldsymbol{x}(d),d)+g(\boldsymbol{x}(d),d)\boldsymbol{u}(d)\\ \label{eq25b}
\boldsymbol{0}_n =&~h(\boldsymbol{x}(d),d).
\end{align}
\end{subequations}
\end{theorem}
\begin{proof}
See\footnote{Note that by virtue of Remark~\ref{remark2}, we do not need \cite[Assumption~3.2]{ref49.2}.} \cite[Theorem~3.8]{ref49.2}.	
\end{proof}
The Partial Differential Equation (PDE) \eqref{eq25a} together with the algebraic equation \eqref{eq25b} is called \emph{regulator equation} and, from Theorem~\ref{th1}, it follows that the solvability of the regulator equation \eqref{eq25} is the key condition for the solvability of Problem~\ref{problem1}. 
\subsection{Controller Design}
In this subsection, a novel control scheme is designed for solving Problem~\ref{problem1} and, consequently, achieving Objective~1 in presence of \emph{time-varying} uncontrolled power injections. More precisely, we first analyze the zero dynamics of system  \eqref{eq6} in order to make the regulator equation \eqref{eq25} simpler.
Then, inspired by the output regulation control theory \cite{ref49.2}, we present the proposed control scheme. 

First, let $\boldsymbol{x}(d)$ in \eqref{eq25}  be partitioned as  $\boldsymbol{x}(d):=\col(\boldsymbol{x}^a(d), \boldsymbol{x}^b(d))$, with $\boldsymbol{x}^a(d):=\col\big(\boldsymbol{\theta}(d),\boldsymbol{\omega}(d)\big)$ and $\boldsymbol{x}^b(d):=\col\big(\boldsymbol{V}(d),\boldsymbol{P}_c(d),\boldsymbol{\delta}(d)\big)$. Then, consider the following PDE:
\begin{equation}\label{eq40}
\frac{\partial \boldsymbol{x}^b(d)}{\partial d}S(d)=\varrho\big(\boldsymbol{x}^b(d),d\big),
\end{equation}
with 
\begin{equation}\label{F}
\begin{split}
\varrho\big(&\boldsymbol{x}^b(d),d\big):=\\
&\left(\begin{array}{c}
\tau_v^{-1}\big(-\chi_d E\big(\boldsymbol{x}^b(d)\big)\boldsymbol{V}(d)+\bar{E}_f\big)\\
\tau_c^{-1}\big(-\boldsymbol{P}_c(d)+u_e^\ast(\boldsymbol{x}(d),d)\big)\\
\tau_{\delta}^{-1}\big(-\boldsymbol{\delta}(d) +\boldsymbol{P}_c(d)-\xi^{-1}QL^{\mathrm{com}}\big(Q\boldsymbol{\delta}(d) +R\big)\big)
\end{array}\right),
\end{split}
\end{equation} 
where $u_e^\ast(\boldsymbol{x}(d),d)$ will be defined in the following theorem.  
Note that in \eqref{F}, we have replaced $E\big(\boldsymbol{\theta}(d)\big)$ by $E\big(\boldsymbol{x}^b(d)\big)$. This follows from recalling that for each $i=1,\dots,n$, the $i$-th output $h_i$ of system \eqref{eq6} has relative degree equal to 2 (see Lemma~\ref{lemma 1}). Then, by considering the output and its first-time derivative being identically zero, $\boldsymbol{\omega}(d)$ and $\boldsymbol{\theta}(d)$ can be expressed as the solutions to
\begin{equation}\label{hhh}
\begin{split}
\boldsymbol{0}_n&=\boldsymbol{\omega}(d)\\
\boldsymbol{0}_n&=-\boldsymbol{P}_c(d)+\Gamma d-\mathcal{A}\Upsilon\big(\boldsymbol{V}(d)\big)\sin\big(\boldsymbol{\theta}(d)\big),
\end{split}
\end{equation}
respectively. 

In the following theorem, we propose a controller solving Problem~\ref{problem1}.
\begin{theorem}\textbf{(Output regulation based controller).}\label{th2}
Let Assumptions~\ref{ass1}--\ref{ass3} hold and suppose that there exists a solution to \eqref{eq40} for $d\in\mathcal{D}$.
Consider system \eqref{eq6} in closed-loop with 
\begin{equation}\label{eq40.2}
\begin{split}
	{u}=&~{u}^\ast_e\big(\boldsymbol{x}(d),d\big)+K_x\big(x-\boldsymbol{x}(d)\big),\\
\end{split}
\end{equation}
where
\begin{equation}\label{eq41}
	\begin{split}
		u_e^\ast\big(\boldsymbol{x}(d),d\big):=&-\tau_c\tau_v^{-1}\mathcal{A}[\sin(\boldsymbol{\theta}(d))]\Upsilon(\boldsymbol{V}(d))\vert \mathcal{A}\vert [\boldsymbol{V}(d)]^{-1}\\
		& \Big(\chi_dE\big(\boldsymbol{\theta}(d)\big)\boldsymbol{V}(d)-\bar{E}_f\Big)+\boldsymbol{P}_c(d) -\tau_c\Gamma S(d),\\
	\end{split}
\end{equation}
 $\boldsymbol{\theta}(d)$ is the solution to \eqref{hhh} and $K_x:=\big(\boldsymbol{0}_{n\times m}~ \boldsymbol{0}_{n\times n}~ \boldsymbol{0}_{n\times n}~ \boldsymbol{0}_{n\times n}~ \mathds{I}_n\big)$. Then, the trajectories of the closed-loop system \eqref{eq6}, \eqref{eq40.2} starting sufficiently close to $(\bar{\theta}, \boldsymbol{0}, \bar{V}, \bar{P}_c, \bar{\delta}, \bar{d})$ are bounded and converge to the set where the frequency deviation is equal to zero, achieving Objective~1. 
\end{theorem}
\begin{proof}
{See Subsection~\ref{Appendix B} in the Appendix.}
\end{proof}	
We notice that in Theorem~\ref{th2}, the controller \eqref{eq40.2} is designed based on the solution to \eqref{eq40}, which we have provisionally assumed to exist. 
Now, we discuss in the following the condition for the solvability of the regulator equation \eqref{eq25}, implying the solvability also of \eqref{eq40}.
\begin{remark}\textbf{(Existence of solution to the regulator equation \eqref{eq25}).}
\label{rm:solutionRE}
The condition for the solvability of the regulator equation \eqref{eq25}, implying the solvability also of \eqref{eq40}, follows from \cite[Corollary~3.27]{ref49.2}, \textit{i.e.}, the solution to \eqref{eq25} exists if all the eigenvalues of the matrix
\begin{equation}\label{hyper1}
A:=\frac{\partial\varrho(x^b,d)}{\partial x^b}\bigg\vert_{(x,d)=(\bar{x},\bar{d})}
\end{equation}
 have nonzero real parts. Furthermore, in Proposition~\ref{proposition2} in Subsection~\ref{Appendix B} in the Appendix, we provide simpler conditions for the solvability of \eqref{eq25}. Indeed, instead of checking all the eigenvalues of the matrix $A$ given in \eqref{hyper1}, we show that it is sufficient to verify some conditions only on some minors of this matrix. 
\end{remark}
\begin{remark}\textbf{(Comparison with} \cite{ref16}\textbf{).}
Note that in this work we consider a class of exosystems describing the dynamics of both the load demand and renewable generation wider than the one considered in \cite{ref16} (see Remark~\ref{remark exo} for more details about the exosystem model). However, although the dimension of our controller is lower than the one in \cite{ref16}, it may be generally more complex. Indeed, the higher complexity seems to be necessary to deal with the nonlinearity of the considered exosystem. We also note that \cite{ref16} achieves OLFC  only for a particular class of exosystems, while  the achievement of OLFC in presence of a wider class of nonlinear \emph{time-varying} uncontrolled power injections appears complex and still challenging. 
\end{remark}

In order to perform also costs minimization besides frequency regulation, we address in the next section the \emph{approximate} OLFC (\mbox{$\epsilon-$OLFC}) problem (see Objective~2).
More precisely, we propose a control algorithm based on an approximate output regulation method, which uses the solution to a PDE that is solved numerically.

\section{Approximate Optimal Load Frequency Control}
\label{sec:AOR}
In the previous section, following the classical output regulation control methodology, we have designed the controller~\eqref{eq40.2} achieving Objective~1. In this section, in order to achieve Objective~2, 
	 we use an approximate output regulation method that solves numerically only the PDE part of the regulator equation and propose a controller based on this solution. 

Consider \eqref{eq6a}, \eqref{eq6b} and the following output
\begin{equation}\label{eq6.}
q(x) = \col\big(\omega, P_c\big).
\end{equation}
 The reference signal is defined as $q^\mathrm{ref}:=\col\big(\boldsymbol{0}_n, P_c^\mathrm{opt}\big)$, where the optimal power generation value $P_c^\mathrm{opt}$ given by \eqref{optimal} is \emph{time-varying} since the uncontrolled power injection $P_d$ is \emph{time-varying} as well.
 Therefore, the tracking error can be defined as
\begin{equation}\label{eq6.0}
e(t) := q(x(t))-q^\mathrm{ref}(d(t)).
\end{equation}
Now, as in the previous section, we first define the nonlinear output regulation problem for system \eqref{eq6a}, \eqref{eq6b} and \eqref{eq6.0} as follows:
\begin{problem}\textbf{(Approximate output regulation).}\label{problem2}
Let the initial condition  $\big(x(0),d(0)\big)$ of system \eqref{eq6a}, \eqref{eq6b} and \eqref{eq6.0} be sufficiently close to the equilibrium point $(\bar{x},\bar{d})$ satisfying \eqref{eq4}. Then, design a state feedback controller
	\begin{equation}\label{uuuu}
	u(t)=k\big(x(t),d(t)\big)
	\end{equation}
	such that the closed-loop system \eqref{eq6a}, \eqref{eq6b}, \eqref{eq6.0} and \eqref{uuuu} has the following properties:
	
	\begin{enumerate}[]
		\item \begin{property}\label{p3}
			The trajectories $\col\big(x(t),d(t)\big)$ of the closed-loop system exist and are bounded for all $t\geq0$;
		\end{property}
		\item \begin{property}\label{p4}
			The trajectories $\col\big(x(t),d(t)\big)$ of the closed-loop system  satisfy $\lim _{t \rightarrow \infty} \Vert e(t)\Vert \leq \epsilon$, for sufficiently small $\epsilon\in\R_{\geq 0}$, achieving Objective~2.
		\end{property}
	\end{enumerate}
\end{problem}

Then, from Theorem~\ref{th1}, the regulator equation \eqref{eq25} for system \eqref{eq6a}, \eqref{eq6b} and \eqref{eq6.0} becomes
\begin{subequations}\label{eq6.1..}
\begin{align}\label{eq6.1..a}
\dfrac{\partial \boldsymbol{x}}{\partial d}S(d) =&~f(\boldsymbol{x}(d),d)+g(\boldsymbol{x}(d),d)\boldsymbol{u}(d) \\ \label{eq6.1..b}
\boldsymbol{0}_{2n} =&~q(\boldsymbol{x}(d))-q^\mathrm{ref}(d),
\end{align}
\end{subequations}
and the solvability of~\eqref{eq6.1..} implies the solvability of Problem~\ref{problem2}.



In the remaining of this section, we compute the solution to (\ref{eq6.1..a}) numerically and present an algorithm that uses the solution to (\ref{eq6.1..a}) to solve Problem~\ref{problem2} via the minimization of a penalty function (also called performance measure) depending on the tracking error \eqref{eq6.0}. Note that a similar problem, \textit{i.e.}, the design of a controller based on the output regulation theory ensuring that the penalty function is arbitrarily close to its infimum has been studied in \cite[Problems 4, 5]{twente} and it is referred to as \emph{suboptimal} output regulation.
Before introducing the algorithm, we note that the existence of a solution to \eqref{eq6.1..a} corresponds to the stability of the equilibrium point $\overline x$ of \eqref{eq6a} with constant exogenous inputs $\overline d$ such that $(\overline x, \overline d)$ satisfies \eqref{eq4} (see \cite[Theorem~4]{ref49.05} and the center manifold theorem \cite[Theorem~2.25]{ref49.2} for more details). Since, by virtue of Remark~3, the state-feedback controller $u=K_x x$ makes the closed-loop system asymptotically stable, then the solution to  (\ref{eq6.1..a}) exists. 

Now, similarly to \cite[Theorem~5]{ref49.05}, in the following proposition, a penalty function is introduced and the relationship between this penalty function and the tracking error \eqref{eq6.0} is investigated, showing how the numerical solutions to the PDE (\ref{eq6.1..a}) can be used for designing a controller satisfying (\ref{eq6.1..b}) when $t$ approaches infinity, solving Problem~\ref{problem2}. 
\begin{proposition}\textbf{(Approximate output regulation based controller).}\label{proposition1}
Let the {compact} set $\Lambda\ni 0$  and $ \bar{\epsilon}_0 , \bar{\epsilon}_1\in\R_{\geq 0}$ exist such that for all $ \bar{\epsilon}\in [\bar{\epsilon}_0 ,~\bar{\epsilon}_1]$ and $d\in\Lambda$, there exist sufficiently smooth functions $\boldsymbol{x}_{\bar{\epsilon}}(d)$ and $\boldsymbol{u}_{\bar{\epsilon}}(d)$ satisfying
{\begin{equation}\label{eq6.2}
I\big(\boldsymbol{u}_{\bar{\epsilon}}(d)\big):=\int_{\Lambda}\big\Vert q(\boldsymbol{x}_{\bar{\epsilon}}(v))-q^\mathrm{ref}(v)\big\Vert ^2~dv_1\dots dv_{n(n_d+1)}=\bar{\epsilon},
\end{equation}}
where 
 $\big(\boldsymbol{x}_{\bar{\epsilon}}(d), \boldsymbol{u}_{\bar{\epsilon}}(d)\big)$ represents the  solution to (\ref{eq6.1..a}). If the control input is designed as
\begin{equation}\label{law2}
	u=\boldsymbol{u} _{\bar{\epsilon}}(d)+K_x(x-\boldsymbol{x}_{\bar{\epsilon}}(d)),
\end{equation}
where $K_x$ is as in \eqref{eq40.2}, then, there exist $\alpha,  \beta, c \in\R_{>0}$ such that
\begin{equation}\label{eq6.3}
 \Vert e(t)\Vert \leq \alpha e^{-c t}+\beta \bar{\epsilon},~~~~~\forall \bar{\epsilon}\in [\bar{\epsilon} _0,  \bar{\epsilon} _1],
\end{equation}
where $e(t)$ is given by (\ref{eq6.0}).
\end{proposition}
\begin{proof}
{See \cite[Theorem~5]{ref49.05}.}
\end{proof}	
\setlength{\algomargin}{.5em}
\begin{algorithm}[!t]
	\caption{Approximate output regulation}\label{alg:1}
	\DontPrintSemicolon
	\SetArgSty{}
	\SetKwIF{If}{ElseIf}{Else}{if}{}{else if}{else}{end if}
	\SetKwFor{ForAll}{for all}{do}{end forall}
	\SetKwRepeat{Do}{do}{end}
	\textbf{Initialization:} choose $(\boldsymbol{x}_0(d),\tilde{\boldsymbol{u}}_0(d))$ such that $(\boldsymbol{x}_0(d),\tilde{\boldsymbol{u}}_0(d)+K_x\boldsymbol{x}_0(d))$ is solution to \eqref{eq6.1..a};\;\smallskip
	 choose $\bar{\epsilon}$;\;\smallskip
	  set $i \coloneqq 0$;\;\bigskip
	\While{$I(\tilde{\boldsymbol{u}}_i(d))>\bar{\epsilon}$}{\smallskip
	{compute the gradient of \eqref{eq6.2} to obtain $\tilde{\boldsymbol{u}}_{i+1}(d)$;\;\smallskip		
	solve \eqref{eq6.1..a} with $\boldsymbol{u}(d):=\tilde{\boldsymbol{u}}_{i+1}(d)+K_x\boldsymbol{x}_{i+1}(d)$ to obtain $\boldsymbol{x}_{i+1}(d)$;\; \smallskip
	set $i:=i+1$;\; \smallskip		
		}
}
\smallskip
set $\boldsymbol{x}_{\bar{\epsilon}}(d):=\boldsymbol{x}_i(d)$ and $\boldsymbol{u}_{\bar{\epsilon}}(d):=\tilde{\boldsymbol{u}}_i(d) + K_x \boldsymbol{x}_{\bar{\epsilon}}(d)$;\;
\bigskip
\textbf{Output:} The control input is given by \eqref{law2}.
\end{algorithm}

%
%

Then, based on Proposition \ref{proposition1} we use Algorithm \ref{alg:1} (see \cite[Algorithm~1]{ref49.05} for more details) to solve Problem~\ref{problem2}, achieving Objective~2. We refer the reader to \cite[Section~4]{ref49.05} for the details about the convergence proof of Algorithm~1.
Basically, Algorithm~\ref{alg:1} seeks to find iteratively a controller through solving  \eqref{eq6.1..a} numerically. Since both the frequency regulation and costs minimization are taken into account in the penalty function (\ref{eq6.2}), the obtained controller achieves Objective~2 with accuracy $\epsilon:=\beta\bar{\epsilon}$.

\begin{remark}\textbf{(Approximation error)}\label{remark6}
		Note that the approximation error $\bar{\epsilon}$ in Proposition~\ref{proposition1} may be due to the numerical inaccuracy of
		the solution $\big(\boldsymbol{x}_{\bar{\epsilon}}(d), \boldsymbol{u}_{\bar{\epsilon}}(d)\big)$ to \eqref{eq6.1..a} or the algebraic part \eqref{eq6.1..b} for unsolvable regulator equations \eqref{eq6.1..} \cite[Remark~6]{ref49.05}. 
		Furthermore, if Algorithm~\ref{alg:1} finds an approximated solution to \eqref{eq6.1..}, \textit{i.e.,} $\bar{\epsilon}>0$, then the controller \eqref{law2} achieves \mbox{$\epsilon-$OLFC} (Objective 2). Otherwise, if Algorithm~\ref{alg:1} finds an exact solution to \eqref{eq6.1..}, \textit{i.e.,} $\bar{\epsilon}=0$, then the  controller \eqref{law2} achieves OLFC. We also note that to achieve an high accuracy, it is sufficient to choose $\bar{\epsilon}$ sufficiently small. Then, we show via an extensive simulation analysis in the next section that the influence of such an error on the performance of the controlled system is negligible in practical applications.
\end{remark}

\begin{remark}\textbf{(Conditions for solvability of Problem~\ref{problem2}).}
Notably, the approximate output regulation control approach presented in Proposition~\ref{proposition1} does not require the solvability conditions provided in Proposition~\ref{proposition2} in Subsection~\ref{Appendix B} in the Appendix. 
Indeed, Algorithm~1 uses only the solution to \eqref{eq6.1..a} to obtain a controller solving Problem~\ref{problem2}, while  the classical output regulation control approach presented in Theorem~\ref{th2} needs the solution to the regulator equation \eqref{eq6.1..}.
\end{remark}

\begin{remark}\textbf{(Comparison with} \cite{ref50.1}\textbf{).}
A constructive approach for solving a linear-convex optimal steady-state problem with constant exogenous disturbances and parametric uncertainty is proposed in \cite{ref50.1}, which includes the OLFC problem as an application. Although this method solves the OLFC problem, the dynamics of the power network are assumed to be linear. Furthermore, the voltage and turbine-governor dynamics are neglected and the exogenous disturbances   (\textit{i.e.}, the uncontrolled power injections) are assumed to be constant. Differently, in this section we have proposed a control approach for nonlinear power networks, achieving approximate optimal load frequency control  (Objective~2) in presence of \emph{time-varying} uncontrolled power injections. 
\end{remark}

\begin{remark}\textbf{(Properties of the controllers \eqref{eq40.2} and \eqref{law2}).} \label{rem9}
{Note that the proposed control schemes \eqref{eq40.2} and \eqref{law2} together with the augmented dynamics \eqref{eq4.6} are fully distributed and require some information about the network parameters and the exosystems, which can be determined in practice from data analysis and engineering understanding.} 
	{Moreover, the PDE \eqref{eq40} or \eqref{eq6.1..} is solved only once and the knowledge of $\delta_i$ requires only local information and information from the neighboring areas, making the proposed control schemes scalable and independent from the network size.}
	{Also, a sensor for the conventional generation  is required at each node  to measure the generated power $P_c$ in order to implement the proposed control schemes \eqref{eq40.2} and \eqref{law2}.} {Furthermore, the theoretical results we have established in Sections \ref{sec:COR} and \ref{sec:AOR} hold also in the case that the parameters of the swing equations \eqref{eq2}, \eqref{eq4.1} and the exosystems dynamics \eqref{eq4.2} are different from one area to another.}
\end{remark}
%

\begin{figure}
\begin{small}
\begin{tikzpicture}[>=stealth',shorten >=1pt,auto,node distance=3.0cm,
                    semithick]
  \tikzstyle{every state}=[circle,thick,draw=black,fill=black!3,text=black]
   \hspace{1.4cm}
  \node[state] (A)                    {Area 1};
  \node[state]         (B) [above right of=A] {Area 2};
  \node[state]         (D) [below right of=A] {Area 4};
  \node[state]         (C) [below right of=B] {Area 3};

  \path[-] (A) edge             node {$B_{12}=28.1$} (B)
  		(D) edge	     node {$B_{14}=22.8$} (A)
           (B) edge              node {$B_{23}=30.7$} (C)
           (C) edge              node {$B_{34}=17.9$} (D);

  \path[<->] (A) edge [bend left, dashed, blue]          node {} (D)
             (C) edge [bend left, dashed, blue]         node {} (B)
           (C) edge [dashed, blue]          	node {} (A)
           (D) edge [bend left, dashed, blue]          	node {} (C);
\end{tikzpicture}
\caption{Scheme of the considered power network partitioned into four areas, where the solid and dashed lines represent the physical and communication networks, respectively. }
\label{f1}
\end{small}
\end{figure}
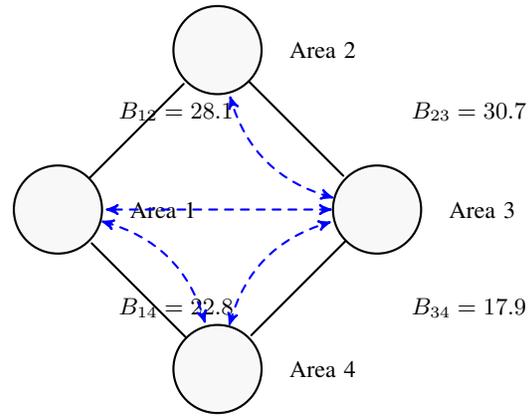

\begin{table}
\caption{System parameters} 
\centering 
\begin{tabular}{c p{1cm} p{1cm} p{1cm} p{1.2cm}} 
\toprule
Parameter  & Area 1 &  Area 2 &  Area 3  & Area 4 \\ [0.5ex]

\midrule 
$B_{ii}$ (p.u.)~         & -56.3       & -58.5      & -56.2      & -49.4 \\
~$q_i~ (\frac{\$10^4}{h})$       & 0.95        & 0.85            & 1.2      & 0.92 \\
$\tau_{v i}$ (s)~ ~~        & 6.32        & 6.63     & 7.15      & 6.46 \\
$X_{di}$ (p.u.)       & 1.76          & 1.81     & 1.87      & 1.91 \\
$X'_{di}$ (p.u.)       & 0.27    & 0.17     & 0.23      & 0.35 \\
 $E_{f_i}$(p.u.)       & 3.85         & 4.43     & 3.96      & 3.88 \\ 
$\tau_{p i}$ (p.u.)~     & 3.95         & 4.71     &5.23      & 4.17 \\
$\psi_i$ (p.u.)~~     & 1.82         & 1.61     & 1.33      & 1.55 \\ 
$\tau_{c i}$ (s)~ ~~      & 7.2         & 6.8     & 8.9      & 7.8 \\ 
$\tau_{\delta i}$ (s)~ ~~     & 0.23     & 0.23     & 0.23      & 0.23 \\
~~~~~~$\xi_i$ (Hz p.u.$^{-1}$)   & 0.73     & 0.73     & 0.73      & 0.73 \\ 
$\kappa_{0 i}$ (p.u.)~     & -8.78      & -8.82      & -8.69       & -8.58  \\ 
$s_{0 i}$ (p.u.)~     & 0.23      & 0.24     & 0.25       & 0.21 \\
$h_i$ (p.u.)~     & 9.63      & 9.71     & 9.59     & 9.68 \\ [1ex] 
\bottomrule 
\end{tabular}
\label{table2}
\end{table} 

\newlength\figureheight 
\newlength\figurewidth

\begin{figure}[t]
		\centering
		\includegraphics[width=\columnwidth]{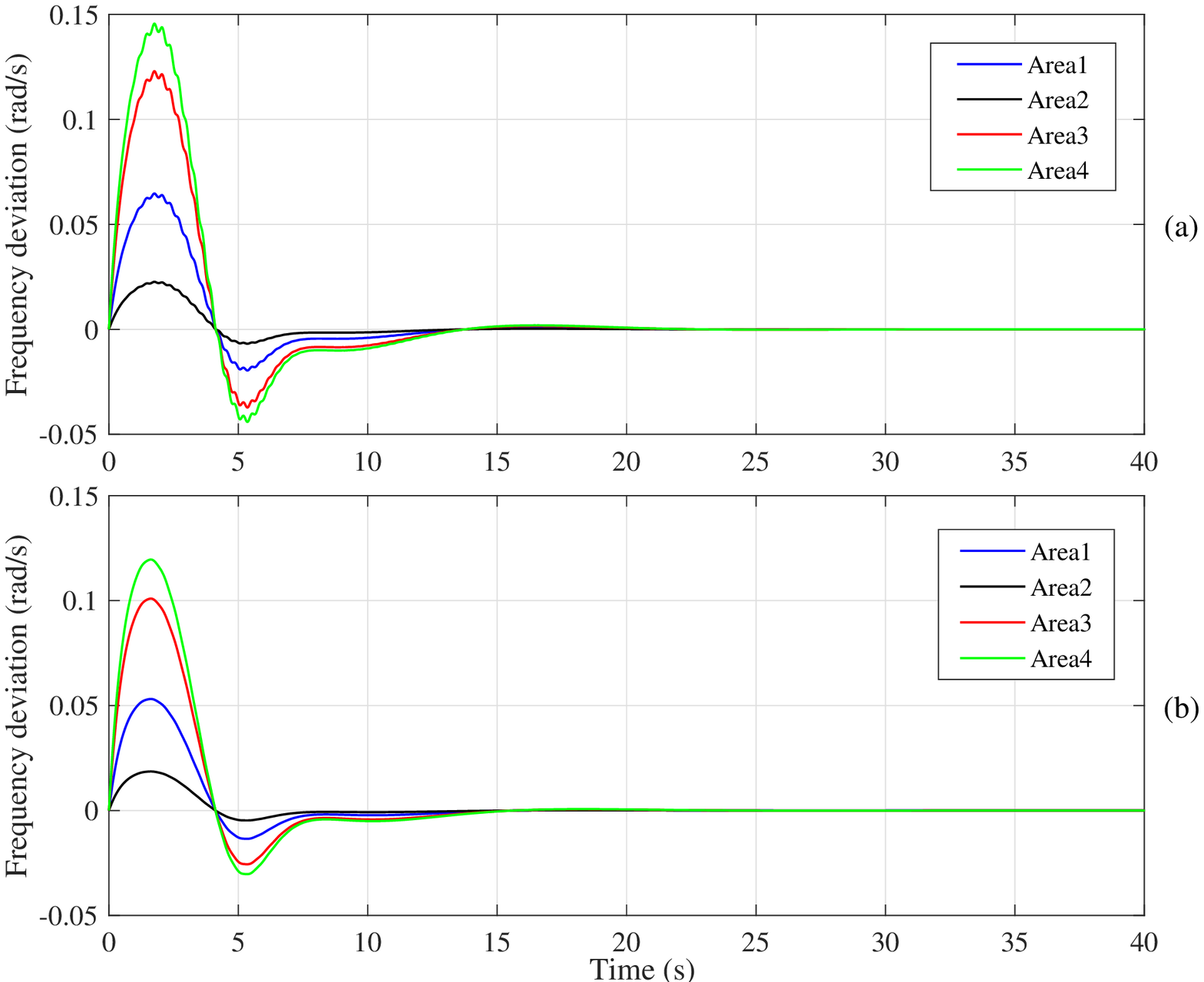}
		\caption{{Scenario~1. Frequency deviation: (a) classical output regulation control approach;
			 (b) approximate output regulation control approach.}}
		\label{f2}
	\end{figure}

\begin{figure}[t]
	\centering
	\includegraphics[width=\columnwidth]{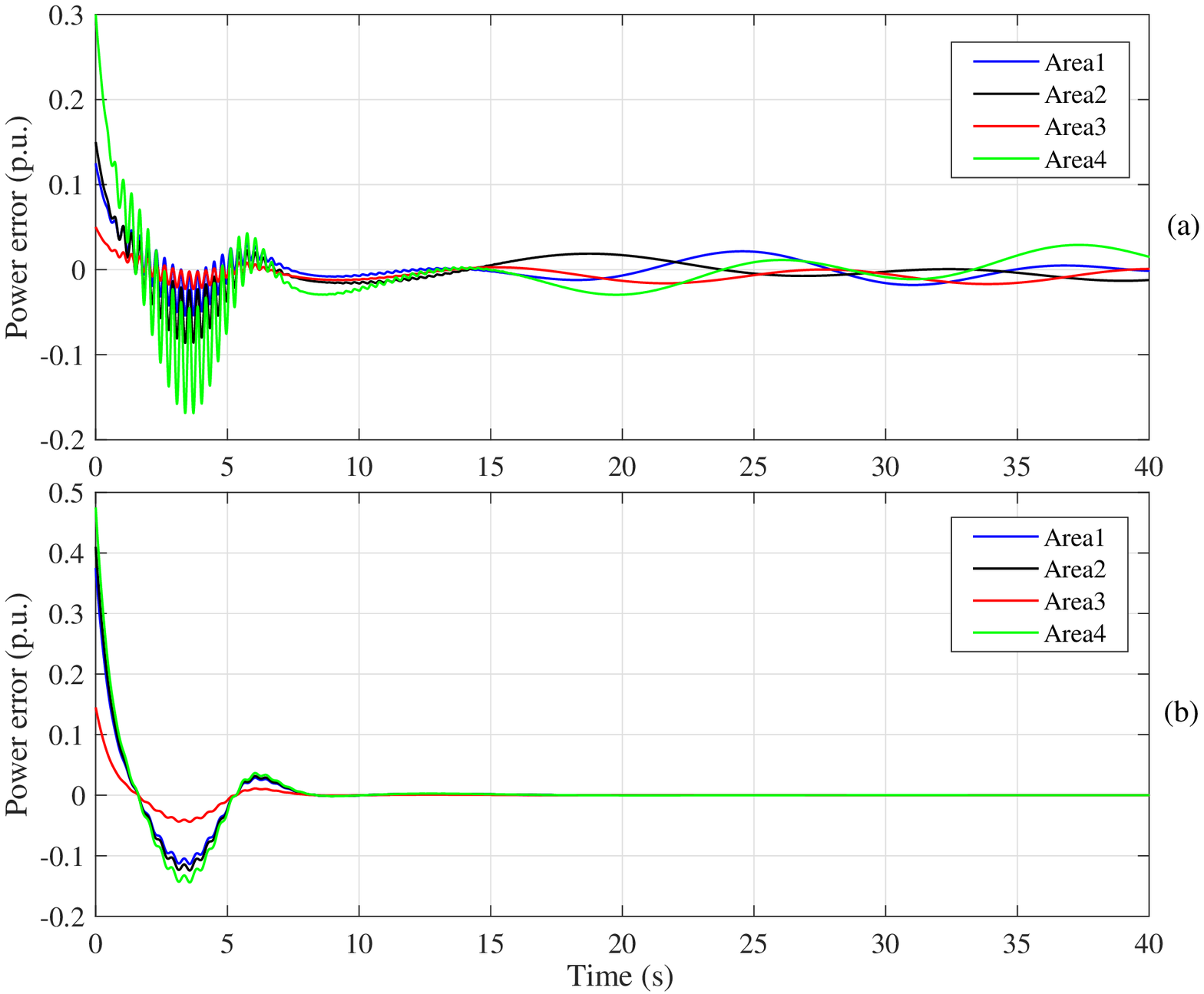}
	\caption{{Scenario~1. Power error: (a) classical output regulation control approach; (b) Approximate output regulation control approach.}}
	\label{f3}
\end{figure}

\begin{figure}[t]
	\centering
	\includegraphics[width=\columnwidth]{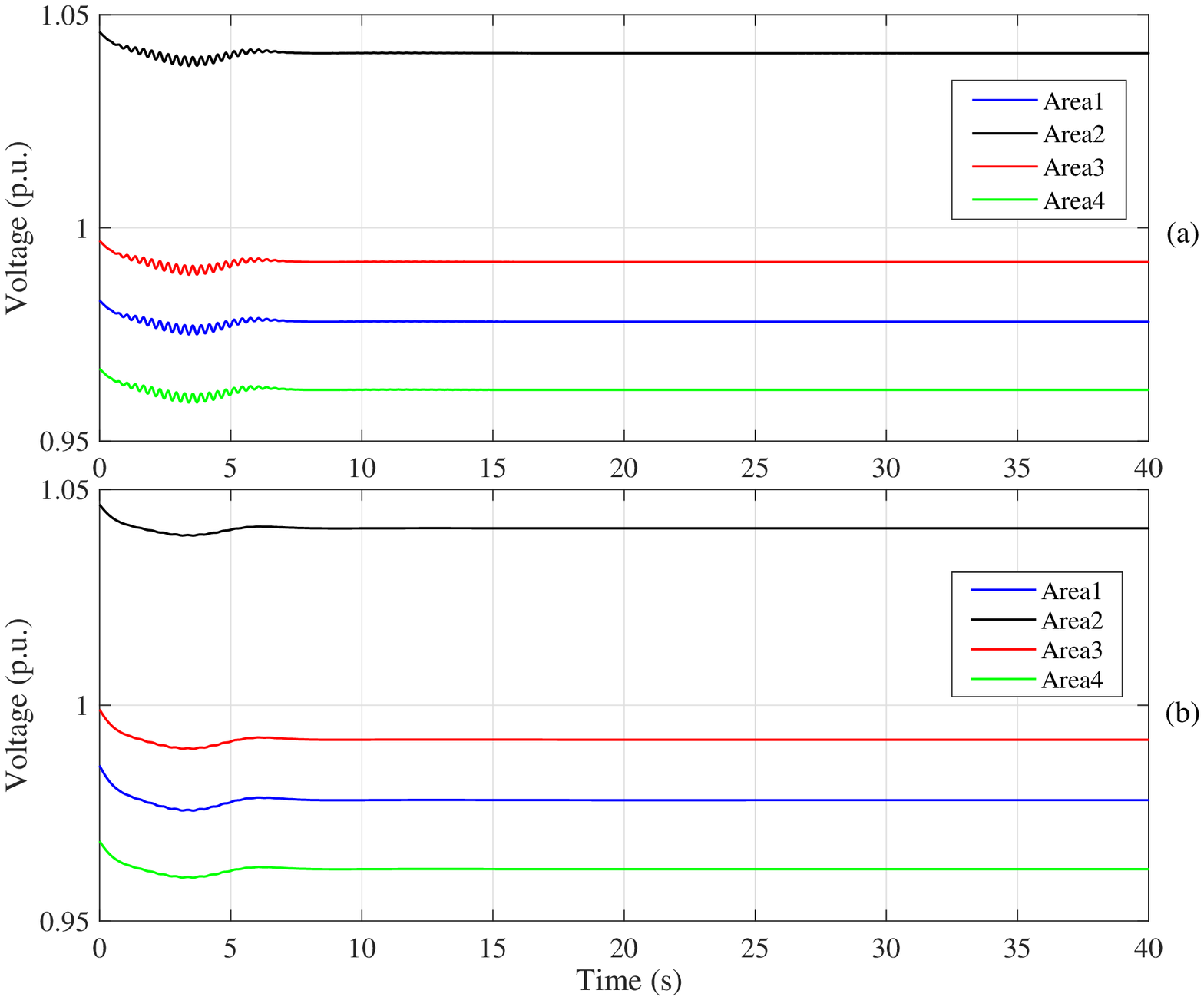}
	\caption{{Scenario~1. Voltage: (a) classical output regulation control approach;  (b) approximate output regulation control approach.}}
	\label{f4}
\end{figure}

\begin{figure}[t]
	\centering
	\includegraphics[width=\columnwidth]{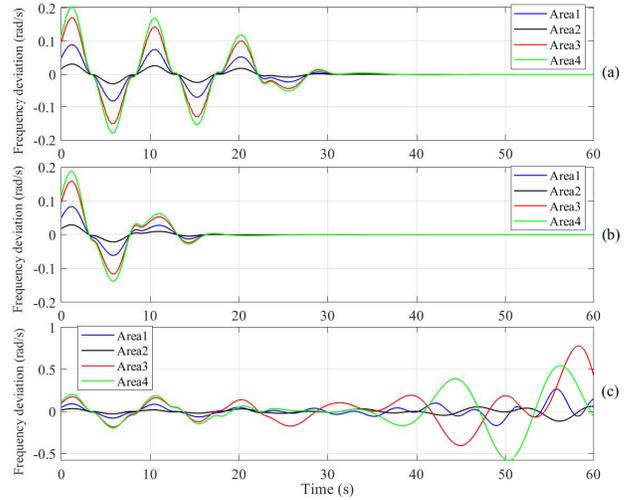}
	\caption{{Scenario~1 with initial conditions not sufficiently close to the desired equilibrium. Frequency deviation: (a) classical output regulation control approach;
			(b) approximate output regulation control approach; (c) linearization method.}}
	\label{flinear}
\end{figure}

\begin{figure}[t]
	\centering
	\includegraphics[width=\columnwidth]{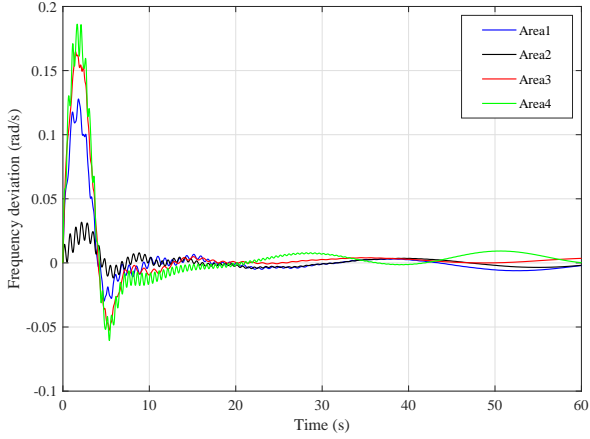}
	\caption{{Scenario~1. Frequency deviation: controller proposed in \cite{ref16}.}}
	\label{fpersis}
\end{figure}

  \begin{figure}[t]
	\centering
	\includegraphics[width=\columnwidth]{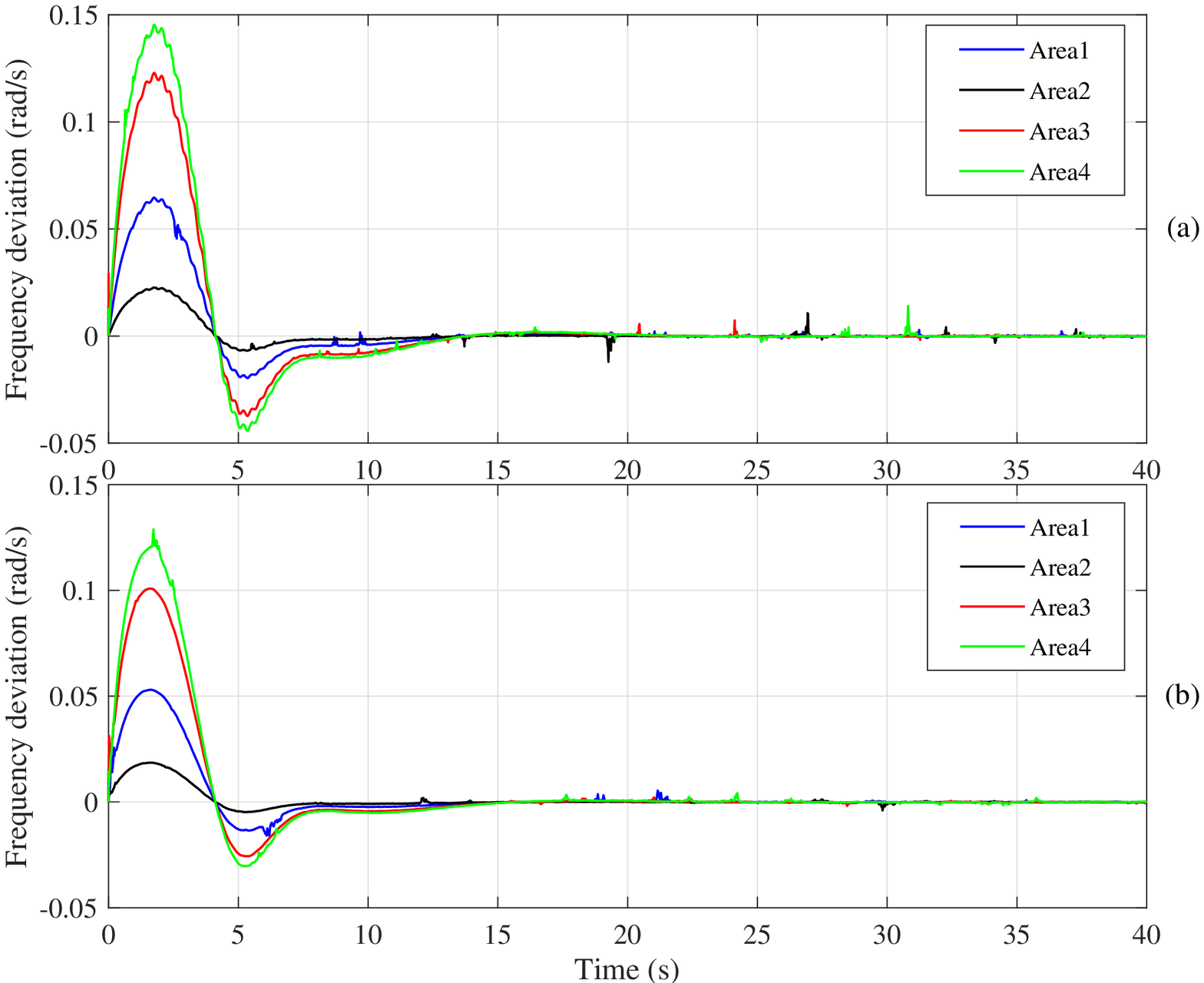}
	\caption{{Scenario~2. Frequency deviation: (a) classical output regulation control approach; (b) approximate output regulation control approach.}}
	\label{fnoise4}
\end{figure}

	\begin{figure}[t]
	\centering
	\includegraphics[width=\columnwidth]{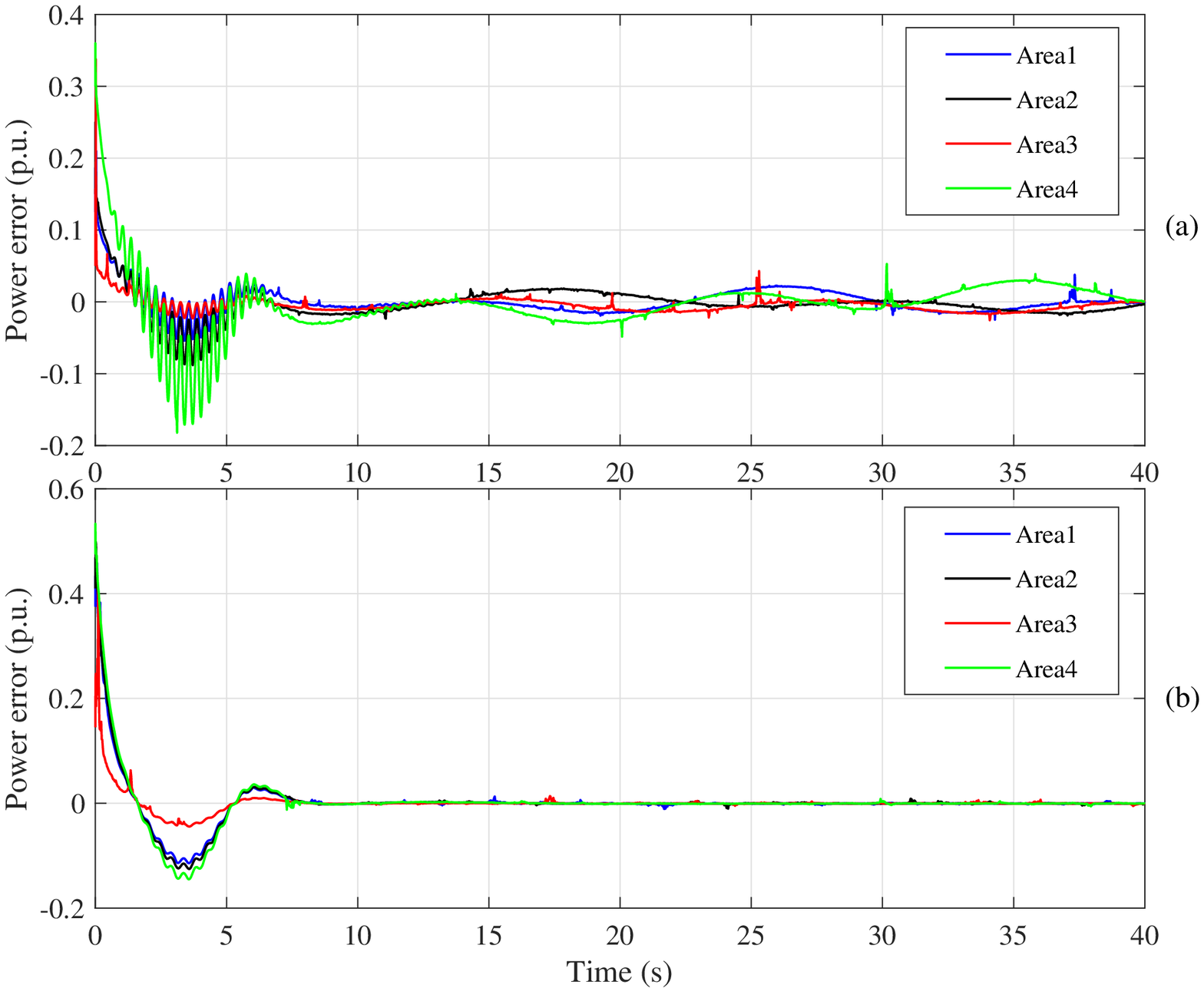}
	\caption{{Scenario~2. Power error: (a) classical output regulation control approach; (b) approximate output regulation control approach.}}
	\label{fnoise5}
\end{figure}

\begin{figure}[t]
	\centering
	\includegraphics[width=\columnwidth]{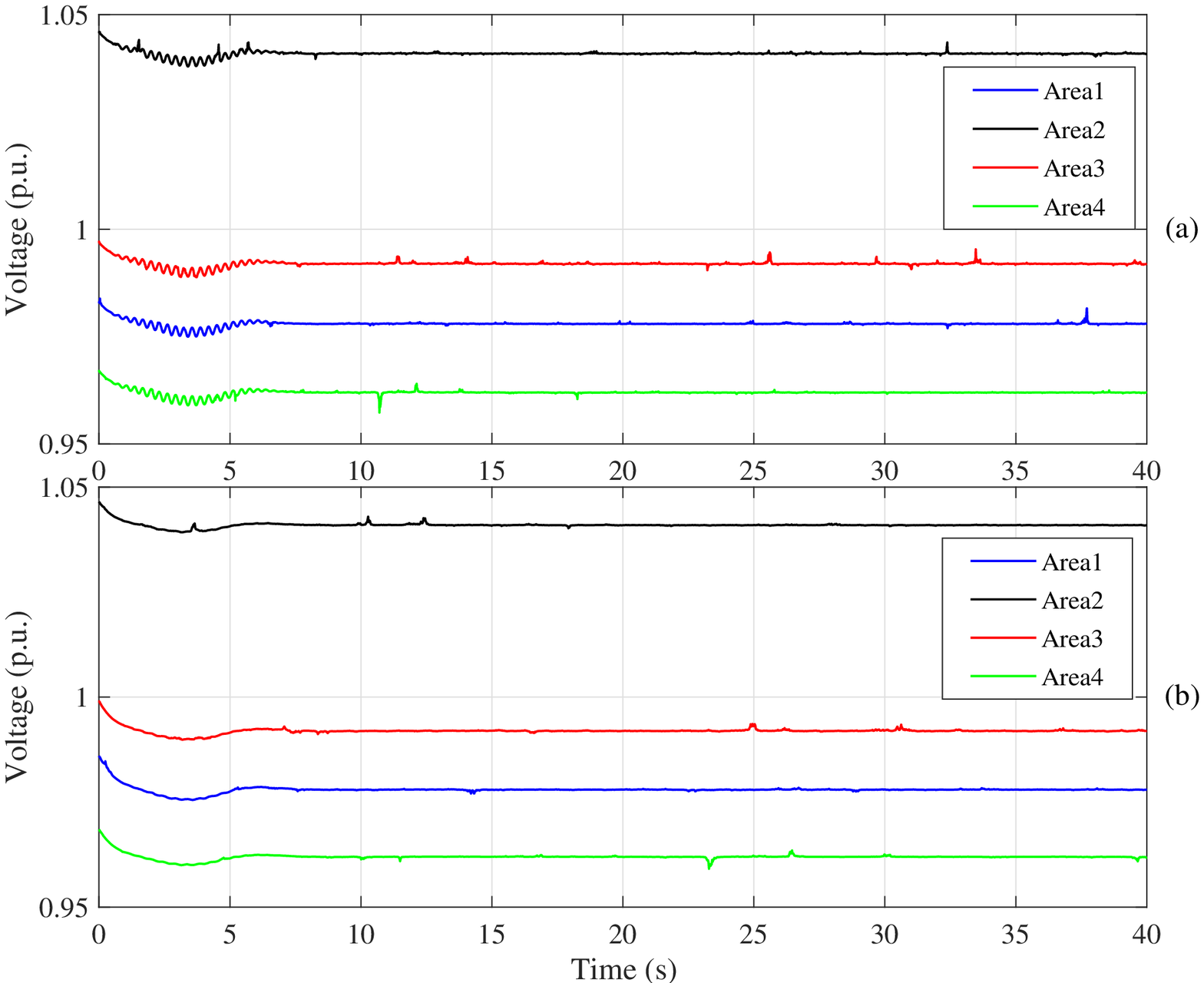}
	\caption{{Scenario~2. Voltage: (a) classical output regulation control approach; (b) approximate output regulation control approach.}}
	\label{fnoise6}
\end{figure}

\begin{figure}[t]
	\centering
		\includegraphics[width=\columnwidth]{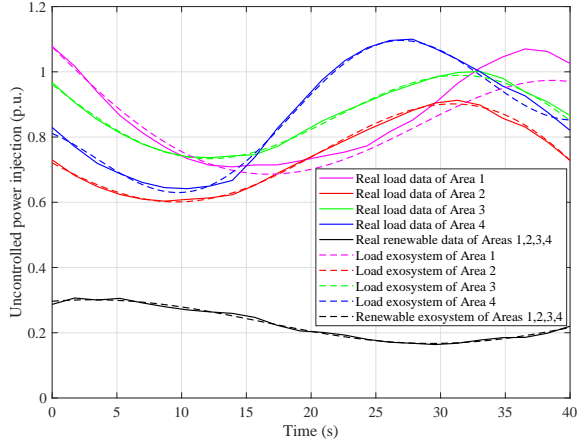}
	\caption{{Scenario~3. Comparison between the absolute value of the real power injections in \cite{data} and the ones produced by the considered exosystems.}}
	\label{f12}
\end{figure}

\begin{figure}[t]
	\centering
	\includegraphics[width=\columnwidth]{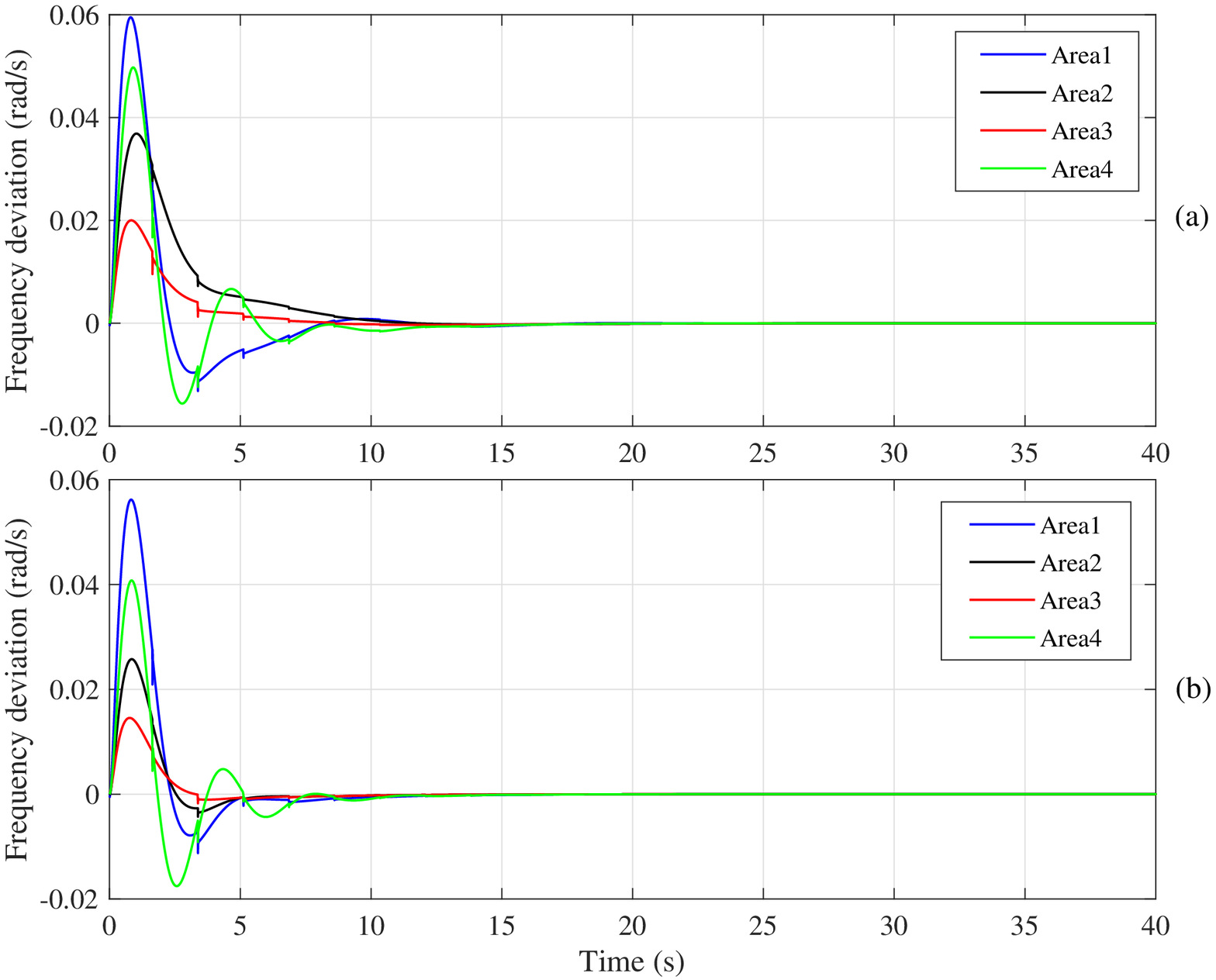}
	\caption{{Scenario~3. Frequency deviation: (a) classical output regulation control approach; (b)  approximate output regulation control approach.}}
	\label{f13}
\end{figure}

\begin{figure}[t]
	\centering
	\includegraphics[width=\columnwidth]{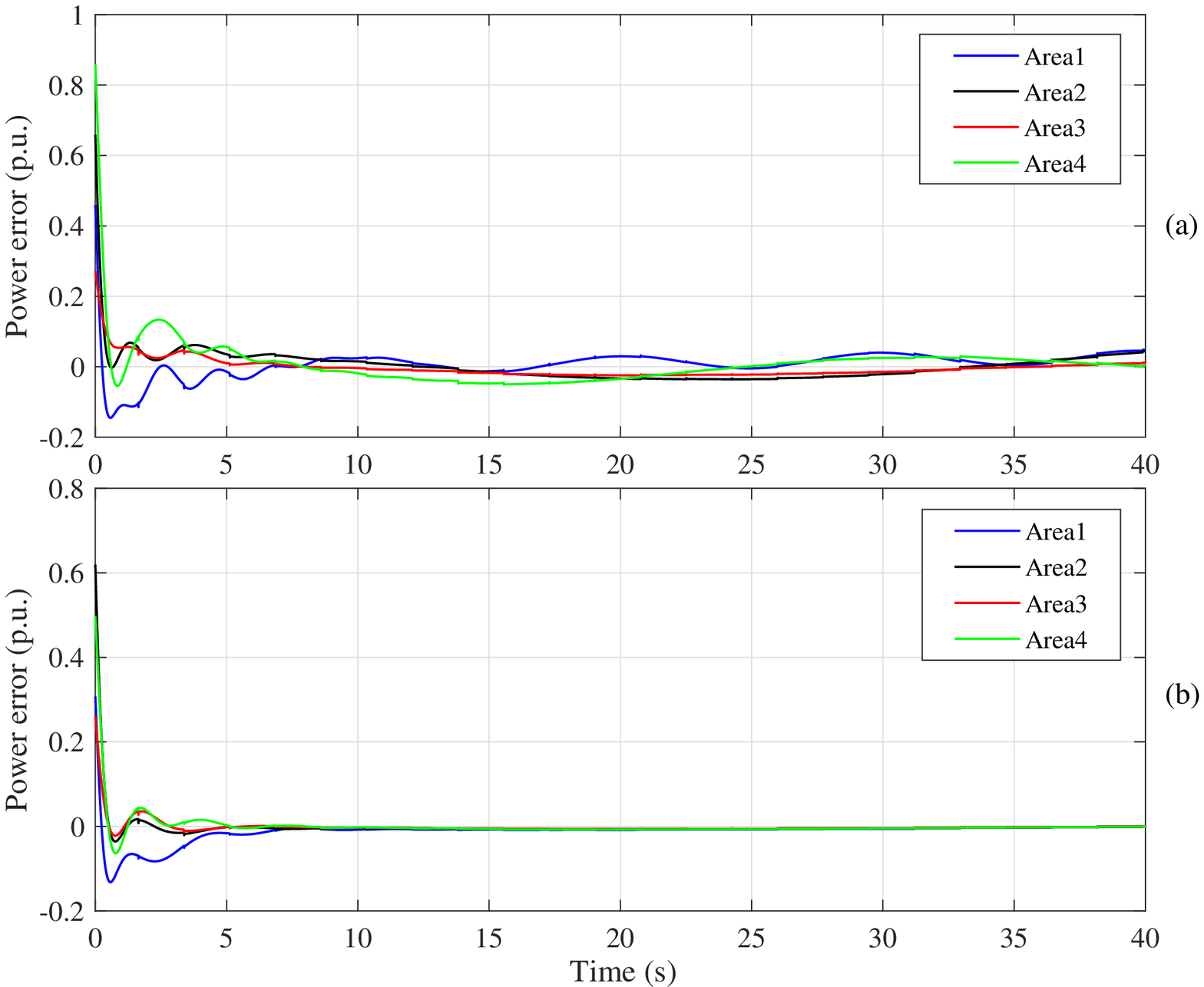}
	\caption{{Scenario~3. Power error: (a) classical output regulation control approach; (b)  approximate output regulation control approach.}}
	\label{f14}
\end{figure}

\begin{figure}[t]
	\centering
	\includegraphics[width=\columnwidth]{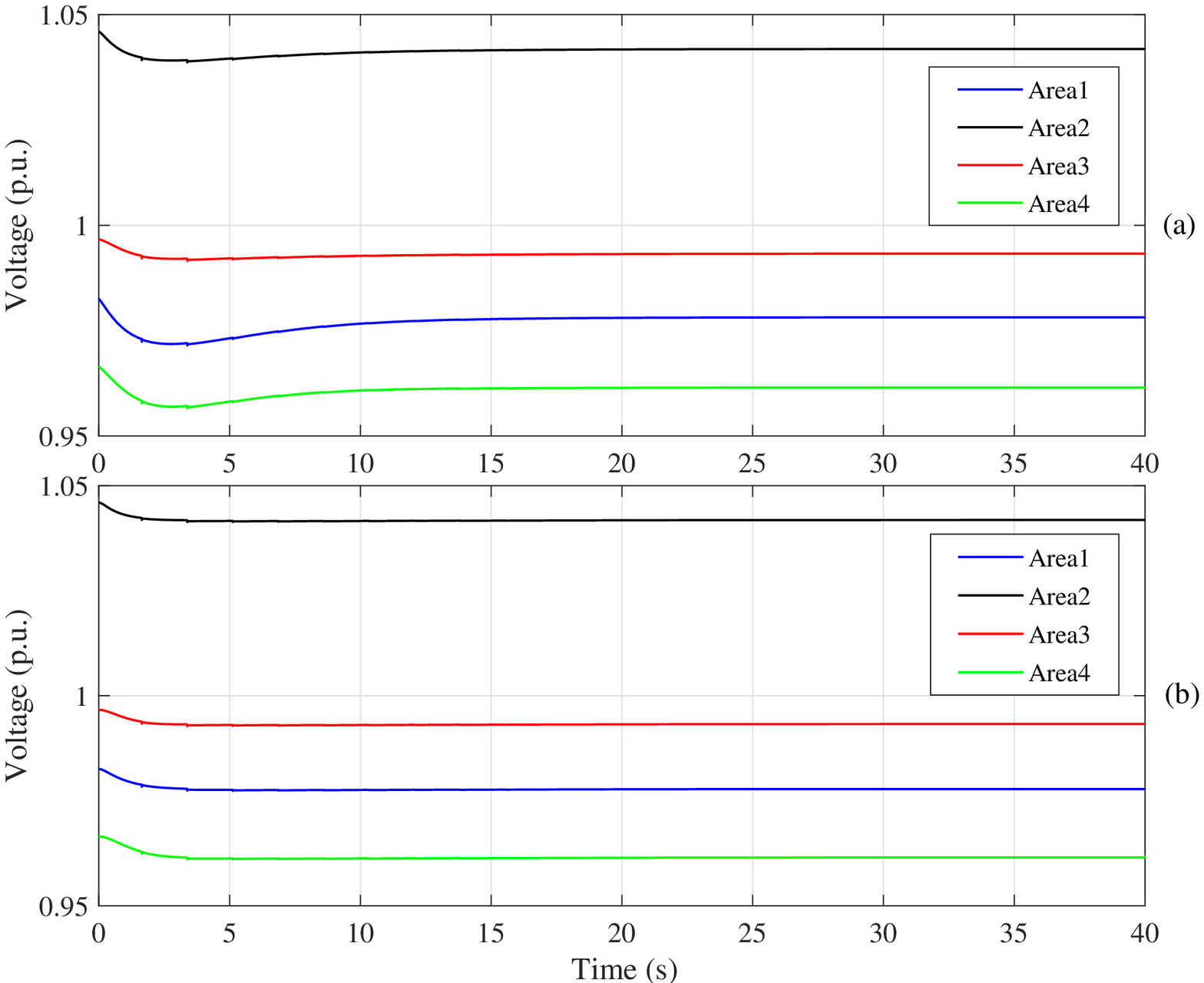}
	\caption{{Scenario~3. Voltage: (a) classical output regulation control approach; (b)  approximate output regulation control approach.}}
	\label{f15}
\end{figure}

%

\section{Simulation Results}
\label{sec:sim}

In this section, an extensive simulation analysis shows excellent performance of the proposed control schemes in three different and critical scenarios. More precisely, we consider a power network partitioned into four control areas (see for instance \cite{ref56.5} on how the IEEE New England 39-bus system can be represented by a network consisting of four areas), which are interconnected as represented in  Fig.~\ref{f1}.
 We assume that each control area includes an equivalent synchronous generator, renewable generations and loads. We provide all the system parameters in Table~\ref{table2}, where the nominal frequency and power base are chosen equal to $120\pi$~rad$/$s and $1000$~MVA, respectively. Also, in Algorithm~\ref{alg:1} we choose  $\bar{\epsilon}=$ \num{1e-7}.

\subsection{Scenario~1: Standard Operating Conditions}
For describing the behaviour of the renewable generation source $i\in\mathcal{V}$, we use the dynamical nonlinear model presented in \cite{ref49.1}, \textit{i.e.},
\begin{equation}\label{eq4.2a}
\begin{split}
\dot{z}^w_{0 i}&=0\\
\dot{z}^w_{1 i}&=z^w_{2 i}\big(\kappa_{0 i} \ln (z^w_{2 i})-\kappa_{0 i} h_i +\frac{s_{0 i}^2}{2}\big) \\
\dot{z}^w_{2 i}&=-z^w_{1 i}\big(\kappa_{0 i} \ln (z^w_{1 i})-\kappa_{0 i} h_i +\frac{s_{0 i}^2}{2}\big) \\
P_{w i}&= z^w_{0i}+z^w_{1i},
\end{split}
\end{equation}
where $z^w_{0 i},, z^w_{1 i}, z^w_{2 i} : \R_{\geq 0}\rightarrow \R$ are the state variables and $\kappa_{0 i}, s_{0 i}, h_i\in\R$ are constant parameters that have been identified in \cite{ref49.1}. 
%
%
Moreover, similarly to \cite{ref16}, we describe the behaviour of the load $i\in\mathcal{V}$ by the following dynamical exosystem:
\begin{equation}\label{eq4.2b}
\begin{split}
\dot{z}^l_{0 i}&=0\\
\dot{z}^l_{1 i}&=\dfrac{2\pi}{15} z^l_{2 i} \\
\dot{z}^l_{2 i}&=-\dfrac{2\pi}{15} z^l_{1 i}\\
P_{l i}&= z^l_{0 i}+z^l_{1 i},
\end{split}
\end{equation}
where $z^l_{0 i}, z^l_{1 i}, z^l_{2 i} : \R_{\geq 0}\rightarrow \R$ are the state variables. 
{Note that, both systems \eqref{eq4.2a} and \eqref{eq4.2b} belong to the class of exosystems we consider in \eqref{eq4.3}, satisfying Assumption~\ref{ass3}.}
Hence, the uncontrolled power injection $i\in\mathcal{V}$ is defined as $P_{di}:=P_{w i}-P_{l i}$.

The system is initially at the steady-state with constant uncontrolled power injections. Then, at the initial time instant $t=$ \SI{0}{\second}, the power generated by the renewable sources and absorbed by the loads is given by \eqref{eq4.2a} and \eqref{eq4.2b}, respectively.
Fig.~\ref{f2} shows the frequency deviations when the classical (see Fig.~\ref{f2}a) and  approximate (see Fig.~\ref{f2}b) output regulation  control approaches are applied, respectively. We can notice that the frequency deviations converge to zero in both cases. 
Additionally, let $P_e:=P_c-P_{c}^\mathrm{opt}$ denote the error between the actual generated power and its corresponding optimal value given by \eqref{optimal}. These errors are shown in Fig.~\ref{f3} when the classical (see Fig.~\ref{f3}a) and approximate (see Fig.~\ref{f3}b) output regulation control approaches are applied, respectively. Although these errors are bounded in both cases, it is evident that the  approximate output regulation control approach achieves Objective~2. {Moreover, as discussed in Remark~\ref{remark6}, we notice that by choosing $\bar{\epsilon}$ sufficiently small ($\bar{\epsilon}=10^{-7}$), the approximate output regulation control approach achieves in practice OLFC. Indeed, by inspecting the time evolution of the norm of the error defined in \eqref{eq6.0}, it appears that its value becomes smaller than \num{1e-3} at $t=$ \SI{16}{\second} and smaller than \num{1e-6} at $t=$ \SI{100}{\second}, implying that Objective 2 is achieved with $\epsilon=$ \num{1e-6} (the plot is not reported due to space limitation). It is then clear that such an error is in practice definitely negligible when affecting the frequency deviation or power error. Also, we have tested the approximate output regulation control approach in presence of constant power injections and by selecting  $\bar{\epsilon}$ equal to zero, achieving OLFC (the plot is not reported due to space limitation).} 
Moreover, we can observe from Fig.~\ref{f4} that also the voltages are stable. 
Finally, we can conclude  that both the control approaches show good performance and guarantee stability. Additionally, the approximate output regulation control approach achieves in practice OLFC. 
{Consider now the case in which the initial conditions of the frequency deviations are not sufficiently close to zero, \textit{i.e.}, $\omega(0) = \col(0.05,0.02,0.09,$ $0.11)$. We can
observe from Figures~\ref{flinear}~(a) and (b)  that by applying the proposed controllers the frequency deviation
at each node converges to zero. On the contrary, we can observe from Figure~\ref{flinear}~(c) that by applying a controller based on the output regulation theory and using for the design the linearization of the considered
nonlinear system around the desired equilibrium, the frequency deviations do not converge to zero.}
{
Finally,  we also compare our controller with the one proposed
in \cite{ref16}, which is designed to deal with linear exosystem models only. We can clearly observe from Figure~\ref{fpersis} that the controller in \cite{ref16} is not capable to achieve frequency regulation due to the nonlinearity of the exosystem \eqref{eq4.2a}.
}

{
\subsection{Scenario~2: Measurement Noise}
We consider Scenario~1 adding white noises to the measurement of the generated power $P_c$. 
We can observe from Fig.~\ref{fnoise4} that both the proposed control approaches regulate the frequency deviations to zero,  preserving the stability of the overall network. Also, we can observe from  Fig.~\ref{fnoise5} that the power errors in the approximate output regulation control approach converge to zero (achieving in practice OLFC) and in the classical output regulation control approach remain stable. Furthermore, we can notice from Fig.~\ref{fnoise6} that the voltages in both control methods are stable as well.}

{
\subsection{Scenario~3: Real Data for Uncontrolled Power Injections}
The system is initially at the steady-state with constant uncontrolled power injections. 
Then, at the time instant $t=$ \SI{0}{\second}, we let the uncontrolled power injections vary according to the real values obtained from the dataset\footnote{{Note that the dataset \cite{data} provides hourly load and renewable generation data. However, given the fast dynamics of our system, it does not make sense to show simulations of 24 hours. Since the real uncontrolled power injection profile looks like a sinusoidal signal, we have then reproduced the same signal (in terms of amplitude) with a higher frequency.}} \cite{data}, (where we use the data of four different areas in the United States, {\textit{i.e.}, CAL, CAR, CENT, and FLA for Areas~1, 2, 3, and 4, respectively,}  {on August 29st, 2020}), while the controller uses the information of suitable exosystems, which we have tuned in order to let them generate uncontrolled power injection trajectories that approximate well but not exactly the real ones (see Fig.~\ref{f12}).
Specifically, we design exosystems that produce the following load:  $P_l=11.88\sin(0.059t+0.89)+11.19\sin(0.063t+3.96)+0.0375$ for Area 1, $P_l=0.814\sin(0.032t+1.27)+0.262\sin(0.121t+3.56)+0.05$ for Area~2, $P_l=0.968\sin(0.016t+1.75)+0.211\sin(0.134t+3.28)+0.0375$ for Area~3, $P_l=1.129\sin(0.011t+0.65)+0.168\sin(0.209t+2.42)+0.0125$ for Area~4, and the following renewable generation:
$P_w=0.19\sin(0.007t+1.22)+0.071\sin(0.117t+1.26)+0.05$ for all the areas. Thus, the exosystem \eqref{eq4.2} can be expressed as}
{ \begin{equation}\label{exosystem_sc4}
\begin{split}
\dot{d}^a_{y i}&=0 \\
\dot{d}^b_{y i}&=\left(\begin{array}{cccc} 0 &-\omega^{\alpha}_{y i} &0  &0\\
\omega^{\alpha}_{y i}   &0  &0 &0\\
0 &0 &0 &-\omega^{\beta}_{y i}\\
0 &0  &\omega^{\beta}_{y i}  &0
\end{array}\right)d^b_{y i} \\
P_{yi}&= \Gamma_{y i}  \col\big(d^a_{y i}, d^b_{y i}\big),
\end{split}
\end{equation}
where ${d}^a_{y i} : \R_{\geq 0}\rightarrow\R$, $d^b_{y i} : \R_{\geq 0}\rightarrow \R^{4}$ are the states of the exosystem, $\omega^{\alpha}_{y i}$, $\omega^{\beta}_{y i}$ are equal to  the frequency of the sinusoidal terms in $P_l$ and $P_w$. Moreover, the elements of the matrix $\Gamma_{y i}$ can be obtained from the amplitude and phase of the sinusoidal terms in $P_l$ and $P_w$, where $y$ denotes $l$ or $w$ in case of load demand or renewable generation, respectively.  Note that, system \eqref{exosystem_sc4} belongs to the class of exosystems we consider in \eqref{eq4.3}.
We can observe from Fig.~\ref{f13}  that, despite the mismatch between the actual uncontrolled power injections and the ones generated by the corresponding exosystems, the frequency deviation at each node converges to zero in both classical and approximate output regulation methods, showing that the controlled system is ISS with respect to such a mismatch. 
 We can also observe from Fig.~\ref{f14} that also in this scenario the approximate output regulation control approach achieves in practice OLFC. Moreover, Fig.~\ref{f15} clearly shows that the voltages are stable as well.}

\section{Conclusion}
In this paper, we have used the output regulation theory for the design and analysis of control schemes for nonlinear power networks affected by \emph{time-varying} renewable energy sources and loads. 
More precisely, based on the classical output regulation theory we have proposed a controller that provably regulates the frequency deviation to zero even in presence of \emph{time-varying} uncontrolled power injections. Then, besides merely controlling the frequency deviation, we have proposed a controller that additionally reduces the generation costs. {Future research includes the modelling of the uncontrolled power injections as
	stochastic differential equations and the use of the Ito calculus framework to tackle
	the problem of OLFC in power networks.}

\begin{appendix}
In this Appendix, we present the proofs of Lemma~\ref{lemma 1}, Theorem~\ref{th2}, and Proposition~\ref{proposition1}. 
\subsection{Proof of Lemma~\ref{lemma 1}}\label{Appendix A}
\begin{proof}
We use the definition of relative degree given in \cite[Definition~2.47]{ref49.2}.
Then, we have
{\begin{equation}\label{eq8}
\begin{split}
L_{g_a} h=&~\big(\boldsymbol{0}_{n\times m}~ \mathds{I}_n~ \boldsymbol{0}_{n\times n}~ \boldsymbol{0}_{n\times n}~  \boldsymbol{0}_{n\times n}~  \boldsymbol{0}_{n\times n(n_d+1)}\big)g_a\\
=&~\boldsymbol{0}_{n\times n} \\
L_{f_a} h=&~\big(\boldsymbol{0}_{n\times m}~ \mathds{I}_n~ \boldsymbol{0}_{n\times n}~ \boldsymbol{0}_{n\times n}~  \boldsymbol{0}_{n\times n}~  \boldsymbol{0}_{n\times n(n_d+1)}\big)f_a\\
=&~\tau_p ^{-1}\Big(-\psi \omega +P_c+\Gamma d -\mathcal{A}\Upsilon(V)\sin(\theta)\Big) \\
L_{g_a}L_{f_a} h
=&~\tau_p^{-1}\tau_c^{-1},
\end{split}
\end{equation}}
 implying that the relative degree for each control  area is equal to 2. 
\end{proof}
\subsection{Proof of Theorem~\ref{th2}}\label{Appendix B}
\begin{proof}
	By virtue of \cite[Theorem~3.26]{ref49.2}, we first compute the following matrix 
	\begin{equation}
	\begin{split}
	G_e(x,d)=\left(\begin{array}{c}
	h(x,d)\\
	L_{f_a}h(x,d) \\
	\end{array}\right),
	\end{split}
	\end{equation}
Then, we notice that the solution to  $G_e(x,d)=\boldsymbol{0}_{2n}$ for system (\ref{eq6}) can be obtained as follows:
	\begin{equation}\label{eq33}
	\begin{split}
	\boldsymbol{0}_n&=\omega^\ast\\
	\boldsymbol{0}_n&=P_c+\Gamma d-\mathcal{A}\Upsilon(V)\sin(\theta^\ast).
	\end{split}
	\end{equation}
	where $\omega^\ast=\boldsymbol{0}_n$ and $\theta^\ast$ denote the solutions to \eqref{eq33}. 
	Thus, there exist the partition $x^{a}:=\col(\theta,\omega), x^{b}:=\col(V,P_c,\delta)$ and a sufficiently smooth function
	\begin{equation}
	\begin{split}
	\zeta(x^{b},d):=\left(\begin{array}{c}
	\theta^\ast\\
	\boldsymbol{0}_n
	\end{array}\right),
	\end{split}
	\end{equation}
	such that $G_e(x,d)\vert_{x^a=\zeta(x^b,d)}=\boldsymbol{0}_{2n}$.
Recalling that for each $i=1,\dots,n$, the $i$-th output $h_i$ of system \eqref{eq6} has relative degree equal to 2 (see Lemma \ref{lemma 1}), we compute the equivalent control input $u_e(x,d)$ by posing the second-time derivate of the output mapping \eqref{eq6c} equal to zero, \textit{i.e.}, 
	\begin{equation}\label{eq34}
	L_{f_{a}}^2h(x,d)+L_{g_{a}}L_{f_{a}}h(x,d)u_e(x,d)=\boldsymbol{0}_n,
	\end{equation}
obtaining the following expression:
	\begin{equation}\label{eq34.0}
	\begin{split}
	u_e(x,d)=&-(\tau_p^{-1}\tau_c^{-1})^{-1}D_a(x,d)\\
	=&~\tau_c\mathcal{A}\Upsilon(V)[\cos(\theta)]\mathcal{A}^\top\omega-\tau_c\tau_p^{-1}\psi ^2\omega\\
	&+\tau_c \tau_p^{-1} \psi \Big(P_c+\Gamma d-\mathcal{A}\Upsilon(V)\sin(\theta)\Big)\\
	&-\tau_c \tau_v^{-1}\mathcal{A}[\sin(\theta)]\Upsilon(V)\vert \mathcal{A}\vert [V]^{-1}\Big(\chi_dE(\theta)V\\
	&-\bar{E}_f\Big)+P_c+\xi^{-1}\omega-\tau_c\Gamma S(d).
	\end{split}
	\end{equation}
{Now, let $u_e^\ast(x,d) := u_e(x,d)\vert _{x^a= \zeta(x^{b},d)}$. By replacing $\omega$ and $\theta$ in \eqref{eq34.0} with $\omega^\ast$ and $\theta^\ast$ given  in \eqref{eq33}, we obtain }
	\begin{equation}\label{ue}
	\begin{split}
	u_e^\ast(x,d)=&~u_e(x,d)\vert _{x^a=\zeta(x^{b},d)}\\
	=&-\tau_c \tau_v^{-1}\mathcal{A}[\sin(\theta)]\Upsilon(V)\vert \mathcal{A}\vert [V]^{-1}\Big(\chi_dE(\theta^\ast)V\\
	&-\bar{E}_f\Big)+P_c-\tau_c\Gamma S(d).
	\end{split}
	\end{equation}
	Then, the zero dynamics of (\ref{eq6})  are given by
	\begin{equation}
	\begin{split}
	\tau_v \dot{V}=&-\chi_d E(\theta^\ast)V+\bar{E}_f\\
	\tau_c \dot{P}_c=&-P_c+u_e^\ast(x,d)\\
	\tau_{\delta}\dot{\delta}=&-\delta +P_c-\xi^{-1}QL^{\mathrm{com}}(Q\delta +R)\\
	\dot{d}=&~S(d),
	\end{split}
	\end{equation}
	which can be rewritten as
	\begin{equation}\label{zero1}
	\begin{split}
	\dot{x}^b&=\varrho({x}^b(d),d)\\
	\dot{d}&=S(d),
	\end{split}
	\end{equation}
{	where $\varrho({x}^b(d),d)$ is given by \eqref{F}.
Now,   we replace ${x}^b$ in \eqref{zero1} with the solution $\boldsymbol{x}^b(d)$ to \eqref{eq40} and define $\boldsymbol{\theta}(d):=\theta^\ast(\boldsymbol{x}^b(d))$.}
	 Thus, according to \cite[Theorem~3.26]{ref49.2}, the solution to the regulator equation \eqref{eq25} can be expressed as follows:
	\begin{equation}\label{solution1}
	\begin{split}
	\boldsymbol{x}(d)=&\left(\begin{array}{c} 
	\boldsymbol{\theta}(d)\\
	\boldsymbol{0}_n\\
	\boldsymbol{x}^b(d)
	\end{array}\right) \\
	\boldsymbol{u}(d)=&~u_e^\ast(\boldsymbol{x}(d),d),
	\end{split}
	\end{equation}
where ${u}^\ast_e\big(\boldsymbol{x}(d),d\big)$ is given by \eqref{eq41}.
	Hence, following Theorem~\ref{th1}, Problem~\ref{problem1} is solvable.
	Then, the state feedback controller \eqref{eq40.2} guarantees that the trajectories of the closed-loop system \eqref{eq3}, \eqref{eq4.6}, \eqref{eq4.3}, \eqref{eq40.2} starting sufficiently close to $(\bar{\theta}, \boldsymbol{0}_n, \bar{V}, \bar{P}_c, \bar{\delta}, \bar{d})$ are bounded and  converge to the set where the frequency deviation is equal to zero, achieving Objective~1.
\end{proof}

{In the proof of Theorem~\ref{th2}, we have provisionally assumed that the solution to \eqref{eq40}  exist. In the following proposition, the condition for the solvability of  the regulator equation \eqref{eq25}, implying the solvability of \eqref{eq40}, is investigated}.
\begin{proposition}\textbf{(Existence of solution to the regulator equation \eqref{eq25}).}\label{proposition2}
	The solution to regulator equation (\ref{eq25}) exists if $A_{11}$ has no zero real part eigenvalue and for all $\rho\in\mathbb{R}$ 
	\begin{equation}\label{eq43} 
	\begin{split} 
	\det\big(A_{22}-j\rho \mathds{I}_n-A_{21}(A_{11}-j\rho \mathds{I}_n)^{-1}A_{12}\big) \neq 0,
	\end{split} 
	\end{equation} 
	where
	\begin{equation}\label{eqwhere} 
	\begin{split} 
	A_{11}&=-\frac{\partial \big(\tau_v^{-1}\chi_d E(\theta^\ast(x^b(d))){V}\big)}{\partial V}\Big\vert_{(x,d)=(\bar{x},\bar{d})} \\
	A_{12}&=-\frac{\partial \big(\tau_v^{-1}\chi_d E(\theta^\ast(x^b(d))){V}\big)}{\partial P_c}\Big\vert_{(x,d)=(\bar{x},\bar{d})} \\ 
	A_{21}&=\tau_c^{-1}\frac{\partial u_e^\ast(x,d)}{\partial V}\Big\vert_{(x,d)=(\bar{x},\bar{d})}\\
	A_{22}&=-\tau_c^{-1}+\tau_c^{-1}\dfrac{\partial u_e^\ast(x,d)}{\partial P_c}\Big\vert_{(x,d)=(\bar{x},\bar{d})} \\
	A_{33}&=-\tau_{\delta}^{-1}\big(\mathds{I}_n+\xi^{-1}QL^{\mathrm{com}}Q\big),
	\end{split}
	\end{equation}
$\theta^\ast(x^b(d))$ is the solution to \eqref{eq33}, $u_e^\ast(x,d)$ is given by \eqref{ue}, and $(\bar{x},\bar{d})$ satisfies \eqref{eq4}.
\end{proposition}
\begin{proof}
We have discussed in Remark~\ref{rm:solutionRE} that the solution to the regulator equation (\ref{eq25}) exists if all the eigenvalues of the matrix  
\begin{equation}\label{eq44}
\begin{split}
A&:=\frac{\partial\varrho(x^b,d)}{\partial x^b}\Big\vert_{(x,d)=(\bar{x},\bar{d})}\\
&=\left(\begin{array}{ccc} A_{11} &A_{12}  &\boldsymbol{0} \\
A_{21} &A_{22} &\boldsymbol{0} \\
\boldsymbol{0} &\tau_{\delta}^{-1}   & A_{33}
\end{array}\right),
\end{split}
\end{equation}
have nonzero real part, where $A_{ij}, i,j=1,2,3$ are defined in~\eqref{eqwhere}. Now, let $\lambda$ denotes the eigenvalues of matrix $A$. Then, by using the Schur complement of the block $A_{33}-\lambda \mathds I_{n}$ of the matrix $A-\lambda \mathds{I}_{3n}$, the eigenvalues of $A$ satisfy
\begin{equation}\label{eq47..}
\begin{split}
\det\big( A-\lambda \mathds{I}_{3n}\big)=&
~\det\big(A_{33}-\lambda \mathds{I}_n\big) \cdot\\
&~\det \underbrace{\left(\begin{array}{cc} A_{11}-\lambda \mathds{I}_n &A_{12} \\ A_{21} &A_{22}-\lambda\mathds{I}_n
\end{array}\right)}_{:= \, \tilde A}\\
=&~\det\big(A_{33}-\lambda \mathds{I}_n\big)\cdot\det\big(A_{11}-\lambda \mathds{I}_n\big)\cdot\\
&\det\Big(A_{22}-\lambda \mathds{I}_n\\
&-A_{21}(A_{11}-\lambda \mathds{I}_n)^{-1}A_{12}\Big)\\
=&~0,
\end{split}
\end{equation}
where the second equality is obtained by using again the Schur complement of the block $A_{11}-\lambda\mathds{I}_n$ of the matrix $\tilde A$.
We notice that the matrices $\xi$ and $Q$ are positive definite matrices and $L^{\mathrm{com}}$ is a positive semi-definite matrix. Therefore, $A_{33}$ is a negative definite matrix. Also, by virtue of the assumptions in the Proposition statement, $A_{11}$ has no zero real part eigenvalues and  (\ref{eq43}) holds for all $\rho\in\mathbb{R}$. Then all the eigenvalues of $A$ have nonzero real parts. Consequently, according to \cite[Corollary~3.27]{ref49.2} the solution to the regulator equation (\ref{eq25}) exists.
\end{proof}

\end{appendix}

\ifCLASSOPTIONcaptionsoff
  \newpage
\fi



%

\balance

\end{document}